\label{key}%%%%%%%%%%%%%%%%%%%%%%%%%%%%%%%%%%%%%%%%%%%%%%%%%%%%%%%%%%%%%%%%%%%%%%%%%%%%%%%%
\documentclass[10pt,journal,compsoc]{IEEEtran}
\pdfoutput=1
\IEEEoverridecommandlockouts
\def\BibTeX{{\rm B\kern-.05em{\sc i\kern-.025em b}\kern-.08em
		T\kern-.1667em\lower.7ex\hbox{E}\kern-.125emX}}
\usepackage{cite}
\usepackage{amsmath,amssymb,amsfonts}
\usepackage{algorithmic}
\usepackage{graphicx}
\usepackage{textcomp}
\usepackage{xcolor}
\usepackage{multirow}
\usepackage{graphicx}
\usepackage{multicol}
\usepackage{subcaption}
\usepackage{stfloats}
\usepackage{amsthm}
\usepackage{footnote}

\newtheorem{prop}{Proposition}
\newtheorem{corollary}{Corollary}
\theoremstyle{definition}
\newtheorem{proofpart}{Part}
\newtheorem{theorem}{Theorem}
\makeatletter
\@addtoreset{proofpart}{theorem}
\makeatother

\makeatletter
\renewcommand{\@IEEEsectpunct}{\newline\ \,}% Modified from {:\ \,}
\makeatother

\title{Investment in EV charging spots for parking
}

\author{Brendan Badia, Randall Berry, Ermin Wei% <-this % stops a space
\IEEEcompsocitemizethanks{\IEEEcompsocthanksitem The authors are with the Department
of Electrical Engineering and Computer Science, Northwestern University.}
}

\begin{document}
	
	\maketitle
	%\thispagestyle{plain} %before this was empty for both
	%\pagestyle{plain}

	%%%%%%%%%%%%%%%%%%%%%%%%%%%%%%%%%%%%%%%%%%%%%%%%%%%%%%%%%%%%%%%%%%%%%%%%%%%%%%%%
	\begin{abstract}
	
	As demand for electric vehicles (EVs) is expanding, there is much interest in meeting the need for charging infrastructure, especially in urban areas. One method of adding charging stations is to install them at parking spots. While there are costs to constructing spots with chargers and they preclude regular internal combustion engine (ICE) drivers from using these spots, EV drivers may have a higher valuation for these spots due to their parking and charging demand. We look at two models for how decisions surrounding investment in charging stations on existing parking spots may be undertaken. First, we analyze two firms who compete over installing stations under government set mandates or subsidies. Given the cost of constructing spots and the competitiveness of markets, we find it is ambiguous whether setting higher mandates or higher subsidies for spot construction leads to better aggregate outcomes. Second, we look at a system operator who faces uncertainty on the size of the EV market. If they are risk neutral, we find a relatively small change in the uncertainty of the EV market can lead to large changes in the optimal charging capacity.
	
	\end{abstract}

	\begin{IEEEkeywords}
	Parking, electric vehicles, congestion, game theory
	\end{IEEEkeywords}

	%%%%%%%%%%%%%%%%%%%%%%%%%%%%%%%%%%%%%%%%%%%%%%%%%%%%%%%%%%%%%%%%%%%%%%%%%%%%%%%%
	\section{Introduction}
	
	Though sales of electric vehicles (EVs) are increasing rapidly, there are still barriers to their wide-spread adoption. A survey by AAA found that ``not enough places to charge" was cited by 63\% of consumers unwilling or unsure about purchasing an electric car, making it the most common reason cited \cite{newsroom}. While 80\% of charging currently occurs at home, in many areas (like densely populated urban areas) the availability of at home charging may not be feasible. It is thus important to study ways of making more charging infrastructure available to the general public. 
	
	Federal and state governments have been pursuing multiple strategies to increase charging infrastructure. One method of increasing availability of charging is by inducing private firms to install charging stations at parking spots. For example, the city of Seattle is working with two private firms to install EV charging stations throughout the city in order to induce more EV drivers \cite{seattlemag}. 
	
	Given that these decisions affect not only the EV market but also the market for ICE cars (who are excluded from using designated charging spots), the potential effects of such policies warrants investigation. As far as we know, this is one of the first attempts to explicitly do so.
	
	A variety of work has focused on optimizing EV charging, both in a competitive market and when system operator is looking at global welfare \cite{lee2015electric,mehar2013optimization,luo2017placement,alizadeh2016optimal,liu2013optimal}. In \cite{lee2015electric} the authors analyze the competition between charging stations (not necessary in a parking context) with renewable generation. The effect of EVs on the power grid and road system is studied in \cite{alizadeh2016optimal}, which offers ways system planners can optimize for this. Our work differs from \cite{lee2015electric} by analyzing how firm's investment in parking affects the market for EVs and ICEs and from \cite{alizadeh2016optimal} by focusing on parking rather than travel decisions of EV drivers along road networks.
	
	A key feature of parking is that it is a congestible resource, i.e., the more users that park in a given location, the less desirable it is (on average) for the next user who parks there. Commonly this effect is modeled with a congestion cost term that is increasing in the number of users parking in a given location \cite{arnott2006integrated}. We follow this approach here and adopt a model similar to those used in the literature on competition with congestible resources, e.g.  \cite{nguyen2016cost,acemoglu2007competition,gibbens2000internet,acemoglu2006price,kelly1998rate,fan2011distributed,johari2010congestible}. In particular \cite{nguyen2016cost} most closely resembles our model, as it looks at wireless service providers pricing unlicensed and licensed spectrum where consumers suffer congestion when other consumers use the same band. We will similarly look at a case where two classes of a market (parking for EV and ICE drivers) are priced and where consumers suffer congestion costs with other consumers of the same class, but differ as each class of service can only serve one type of consumer.
	
	There is a rich literature on examining the effect of prices and competition on parking decisions by consumers and the resulting welfare effects \cite{glazer2001parking,arnott1991temporal,qian2013optimal,shoup2006cruising,shoup1999trouble,ayala2012pricing,arnott1995modeling,lindsey2000traffic,arnott2006integrated}. The interaction of parking and road usage congestion is examined in \cite{glazer2001parking}, which shows a road-usage fee in addition to a parking fee is needed to maximize welfare. The work in \cite{arnott1991temporal} builds upon Vickrey's classic bottleneck road congestion model to examine how different pricing options can help maximize welfare by reducing or eliminating congestion. The effect of setting a minimum amount of parking for urban land use, and how this reduces welfare compared to allowing the market to decide the provision of spots is studied in \cite{shoup1999trouble}. However, none of this work looks at situations where both EV and ICE cars interact. 
	
	This paper focuses on the effect of government policy in two cases: in the first incumbent private owners of parking spaces are installing parking spots with charging stations and in the second a monopolist owns all spots but has uncertainty over the size of the EV market. In the first case we consider how two policies affect these firms: namely introducing subsidies for constructing EV charging spots and setting a minimum proportion of spots that should have chargers installed. 
	This portion is built upon our earlier work \cite{badia2018price}.
	We compare the electiveness of these policies under different assumptions on how firms compete and the relative market sizes of the EV and ICE markets. In the second case, we study the impact of stochasticity in the more valuable market (the EV market) on construction decisions by a monopolist.
	
	We assume a monopolist holds a mass of parking spots in an area with many ICE drivers. Due to the potential difficulties in forecasting the size of the future EV market, they have some uncertainty on how large the EV market may be. We will show both the monopolist's expectation and government policy can have a large impact on outcomes.
	
	The rest of the paper is organized as follows. Section \ref{sec:competitivemodel} defines the model for the competitive market; Section \ref{sec:sequentialspotallocation} analyzes the game as defined by this model; Section \ref{sec:competitiveresults} summarizes the results from this model; Section \ref{sec:competitivenumericalcasestudy} analyzes a specific case study; Section \ref{sec:monopolistmodel} defines the model for the monopolist; Section \ref{sec:monopolistpricing} examines the pricing structure in this model; Section \ref{sec:stochasticarrivalcapacity} examines the impact capacity decisions have on the optimal operating point; Section  \ref{sec:monopolistnumerical} studies a numerical case study for this model; Section \ref{sec:conclusion} concludes and looks at potential future work.
	
	\section{Competitive Market: Model} \label{sec:competitivemodel}
	
	\noindent In this section we present our model for firms competing over servicing EV and ICE drivers.  
	
	\subsection{Consumers}
	
	We assume there is a mass of two classes of consumers: drivers who use EVs and those who use ICE. Consumers of both classes are assumed to be non-atomic. There will be two types of parking spots: regular parking spots for ICE drivers and spots equipped for charging EVs. We assume each class of driver only parks in a spot designated for them. This is true in practice for ICE drivers as there are regulations against parking in spots with chargers, while we assume all EV drivers require charge for their cars and so would only consider parking in a spot with a charger equipped.
	
	As is standard in parking models we model the congestion as a function of the proportion of occupied spots \cite{arnott2006integrated}. This reflects that with higher occupancy each driver has to spend more time on average to find a vacant spot. We assume firms cluster spots of the same type together (i.e., those with charging and those without) and so consumers at each firm's parking location suffer congestion due to drivers of the same class but not due to other class of drivers, nor with any at the other firm. In general we want a congestion function for each marginal consumer that is increasing in the quantity of consumers already parked, decreasing in the number of spots, and convex in the volume-capacity ratio  \cite{arnott2006integrated}. For simplicity we will use a linear congestion function, which will be the same for both classes of consumers, i.e., 
	\begin{equation} \label{eq:congestion}
	S(q,N) = \epsilon \frac{q}{N}, 
	\end{equation} 
	\noindent where $q$ is the quantity of that class of consumers already parked, $N$ is the number of spots allocated for that class of driver, and $\epsilon > 0$ is a constant term that represents the congestion cost if all spots are filled. Note, we have not put any constraints on $q$ and thus allow $q > N$. In this case it would mean on average not only are all spots taken but there are consumers cruising around for spots.
	
	We do not model differences in the length of time a consumer may want to park; they simply want access to the ability to park (and will all pay the same price at a given firm to do so as discussed in the firm section below). One can think of this as a mass of drivers going to work, where they all work the same amount of time. Each consumer therefore purchases the same good, which is the ability to park during the work day.	
	
	We will assume for both classes that there is a linear downward inverse demand curve for parking determined by the valuation and size of each market, which we define for EV drivers ($P_e(q_e)$) and ICE drivers ($P_d(q_d)$) as
	\begin{equation} \label{eq:inversedemandEV}
	P_e(q_e) = W_e(1-\alpha q_e), 
	\end{equation}
	\begin{equation} \label{eq:inversedemandICE}
	P_d(q_d) =  W_d(1-\beta q_d). 
	\end{equation}
	The $y$-intercept of the ICE inverse demand curve is $W_d > 0$, and the slope is $\beta > 0$; while the $y$-intercept for the mass of EV drivers is $W_e > 0$, and the slope is $\alpha > 0$. Instead of modeling explicitly the demand EV drivers have for a certain amount of charge, for simplicity we assume that access to charging results in a fixed additional utility over parking. We capture this by assuming $W_e > W_d$, where $W_e$ captures both demand for the parking and the ability to charge at the parking spot. 
	
	\subsection{Firms}
	
	We assume there are two firms (i.e., parking garage owners) that already own a mass of parking spots suitable for ICE drivers. We normalize the number of total spots to be 1 and the number of spots firm $i$ has as $N_i$. Specifically, we define the proportion of spots that firm 1 owns as $N_1 = \delta$, where $0 \leq \delta \leq 1$. This means firm 2 owns $N_2 = 1 - \delta$ spots. Firms are unable to construct new spots, but either due to government policy or their own choice they are able to convert a portion of their spots to include an EV charger at cost $p$. We will call the number of spots designated for EVs at firm $i$ as $N_{ei}$ and the remaining spots as $N_{di}$.
	
	We assume the firms are close together so on average consumers have no preference for one firm over the other. The model could be extended to assume all consumers are attempting to get to a common location (for example a central business district) where one firm may be closer and thus preferred as in \cite{qian2013optimal}. 
	
	We define the value of parking for a marginal driver of each class at each firm $i$ when the firm sets a price $m_i$ ($c_i$) for ICE cars (EVs) and $q_{di}$ ($q_{ei}$) ICE (EV) drivers are parked at firm $i$ as follows: 
	\begin{multline} \label{eq:utilityEV}
	U_{ei}(q_{e1},q_{e2},N_{ei},c_i) = W_e(1-\alpha (q_{e1} + q_{e2})) \\
	- \epsilon \frac{q_{ei}}{N_{ei}} - c_i, 
	\end{multline}
	\begin{multline} \label{eq:utilityICE}
	U_{di}(q_{d1},q_{d2},N_{di},m_i) = W_d(1-\beta (q_{d1} + q_{d2})) \\
	- \epsilon \frac{q_{di}}{N_{di}} - m_i.
	\end{multline}
	Here, the first terms reflect the inverse demands in \eqref{eq:inversedemandEV} and \eqref{eq:inversedemandICE}. The quantity of each class of driver at firm $i$ is a function of the number of spots available for that driver and the price set for them both by firm $i$ and the opposing firm $j$. The number of consumers at each firm in equilibrium is defined using the notion of a Wardrop equilibrium \cite{wardrop1952some}. A Wardrop equilibrium here means that when drivers are parking at both firms, the following two conditions hold: the marginal utilities of the consumers at each garage have to be equated ($U_{ei} = U_{ej}$, $U_{di} = U_{dj}$) and the marginal utilities have be 0 ($U_{ei} = 0$, $U_{ej} = 0$, $U_{di} = 0$, $U_{dj} = 0$). This means that consumers park until the marginal user would gain negative utility from doing so, and sort themselves at the garage at which they get the highest utility. Let $N^{\epsilon}_{ei} =\frac{N_{ei}}{\epsilon}$, $N^{\epsilon}_{di} =\frac{N_{di}}{\epsilon}$, $\mathbf{N_e} = \{N_{e1},N_{e2}\}$,  $\mathbf{N_d} = \{N_{d1},N_{d2}\}$,  $\mathbf{m} = \{m_1,m_2\}$ and $\mathbf{c} = \{c_1,c_2\}$. Solving the system of equations given by the Wardrop Equilibrium, we get the following closed-form solution for the quantities serviced at firm $i$ given the actions of firm $i$ and the opposing firm $j$ when they both set a price such that $q_{ei} > 0 $ and $q_{ej} > 0$:
	\begin{equation} \label{eq:wardropEV}
	q_{ei}(\mathbf{N_e},\mathbf{c}) = \frac{W_e(1+\alpha N_{ej}^{\epsilon}(c_j-c_i) ) - c_i}{\alpha W_e(1+\frac{N_{ej}}{N_{ei}}) + \frac{1}{N_{ei}^{\epsilon}}},
	\end{equation}
	\begin{equation} \label{eq:wardropICE}
	q_{di}(\mathbf{N_d},\mathbf{m}) = \frac{W_d(1+\beta N_{dj}^{\epsilon}(m_j-m_i) ) - m_i}{\beta W_d(1+\frac{N_{dj}}{N_{di}}) + \frac{1}{N_{di}^{\epsilon}}}.
	\end{equation}
	These equations were derived by assuming  that there exist a $q_{ei} \geq 0$ and $q_{ej} \geq 0$ such that the marginal utilities at each firm are 0, which may not always be the case. In particular this will not be true in two scenarios. First, if one firm sets the number of spots for a given type of user to be $0$, as this leads to an infinite amount of congestion for that type of user. If this happens to one firm and not the other, the above solutions lead to one firm with non-zero capacity acting as a monopolist and the other firm servicing a quantity of $0$ for that class of driver as should occur. If they both set the number of spots to be $0$ then both quantities are 0 as is expected. The second scenario is where a firm sets too high a price so that the marginal utility of any driver parking there in equilibrium is negative. We show in Appendix 10.1, that in any unilateral deviation away from the above equilibrium, no firm would ever have such a price and so there is no problem with ignoring this possibility. We can therefore specify the general problem faced by each firm $i$:
	\begin{multline} \label{eq:firmoptimization}
	\underset{N_{ei},N_{di},m_i,c_i}{\text{max}} 
	\Pi_i = m_iq_{di}(\mathbf{N_d},\mathbf{m}) +   c_iq_{ei}(\mathbf{N_e},\mathbf{c}) - pN_{ei} \\
	\text{s.t.} \ \ \ \ \ \ 
	N_{ei} + N_{di} = N_i, \ N_{di} \geq 0, \ N_{ei} \geq 0. 
	\end{multline}
	Note that each firm seeks to solve such a problem, and that their decisions are coupled through the resulting Wardrop equilbrium quantities. Also note that we set no constraint for prices. This is because any negative price will result in negative profit, while a price that is too large (for example, $m_i > W_d$) will result in no profit. These choices are therefore dominated by setting the price to be $0$ and so there is an inferred constraint due to the nature of the optimization problem. By not explicitly stating this constraint we will be able to evaluate first-order conditions without having to check for corner solutions.
	
	\subsection{Government}
	
	We assume the government wants to induce more charging infrastructure to be constructed by firms. We will examine two methods of doing so: issuing a mandate and subsidizing construction. These are policies being currently explored and implemented. 
	
	For the mandate, we define the proportion of spots mandated for EV charging as $r$ and incorporate this into \eqref{eq:firmoptimization} by adding an additional constraint	
	\begin{equation} \label{eq:governmentmandate}
	N_{ei} \geq r N_i. 
	\end{equation}
	For the subsidy, we will allow the government to set a subsidy $s$ such that the cost of constructing EV spots for firms as a function of the intrinsic cost $t$ is
	\begin{equation} \label{eq:governmensubsidy}
	p = t - s. 
	\end{equation}
	
	\subsection{Consumer Surplus and Total Welfare}
	
	Given that consumers are non-atomic, then the consumer surplus for each class of consumer is given by integrating \eqref{eq:utilityEV} and \eqref{eq:utilityICE} over the quantity serviced at both firms, i.e., 
	\begin{multline} \label{eq:consumerwelfareEV}
	CS_{EV} =  \int_{0}^{q_{e1} + q_{e2}} W_e(1-\alpha q) dq \ - \\ \int_{0}^{q_{e1}}(\epsilon \frac{q}{N_{e1}} + c_1)dq - \int_{0}^{q_{e2}}(\epsilon \frac{q}{N_{e2}} + c_2)dq,
	\end{multline}	
	\begin{multline} \label{eq:consumerwelfareICE}
	CS_{ICE} =  \int_{0}^{q_{d1} + q_{d2}} W_d(1-\beta q) dq \ - \\ \int_{0}^{q_{d1}}(\epsilon \frac{q}{N_{d1}} + m_1)dq - \int_{0}^{q_{d2}}(\epsilon \frac{q}{N_{d2}} + m_2)dq.
	\end{multline}
	Note here we are assuming that drivers park sequentially so that early drivers incur a lower congestion cost than later drivers\footnote{Alternatively, we could assume that all drivers park in ``steady-state" in which case all drivers would see the same congestion cost.}. Each firm's welfare will simply be the total profit they receive given by the objective function in \eqref{eq:firmoptimization}. 
	
	The government has to pay for any subsidies they offer firms. We define the cost of subsidies as
	\begin{equation*}
	G(N_{e1},N_{e2},s) = s(N_{e1} + N_{e2}).
	\end{equation*}
	We therefore define the total welfare as the sum of consumer surpluses, the  firm profits, and the government's cost to provide subsidies, i.e.,
	\begin{equation} \label{eq:welfare}
	W_{total} = CS_{EV} + CS_{ICE} + \Pi_1 + \Pi_2 - G(N_{e1},N_{e2},s) .
	\end{equation}
	
	\section{Competitive Market: Sequential Spot Allocation and Pricing Game} \label{sec:sequentialspotallocation}
	
	We model the overall market interactions as a one shot, two-stage game. In the first stage, the firms simultaneously set the number of parking spots they will convert to have chargers. In the second stage they set prices having committed to the number of installations they made in the first stage. We will solve this game using backward induction, so we first examine the second stage (where parking spot quantities have already been set).
	
	\subsection{Second stage: Price competition}
	
	We will consider three different ways in which firms may set prices given fixed quantities of parking spaces.  These will differ in the degree to which firms price discriminate and anticipate the impact of EV parking. We begin with a ``two pricing case" in which firms discriminate in the prices that they charge  EV and ICE drivers accounting fully for the different preferences of these two classes. However, around 50\% of public chargers are free and thus it seems many private owners of charging stations do not specifically set a price to use that service \cite{plugincarschargenetworks}. We analyze two different scenarios where this might occur. The first is where they set a single price for both classes of drivers, but optimize that price understanding EV drivers' valuations. We call this the ``optimal single price case". In the second case, we assume that the garages had in the past (before constructing EV spots) set prices to compete over only ICE drivers and then maintain this price even after charging stations are added to parking spots. We call this the ``Naive" single price case. Each of these pricing schemes could also arise in part due to regulation that limits the prices that can be charged. We look at each case separately. A summary of the cases is shown in Table \ref{tab:firmpricedecisions} below.  \\
	\begin{table}[h!] 
		\begin{center}
			\begin{tabular}{ |c|c| } 
				\hline
				Case & Constraint on Prices ($c_i$, $m_i$)  \\ \hline
				Two Price & None \\  \hline
				Optimal single price & $c_i = m_i$  \\  \hline
				``Naive" single price & set assuming $N_{e1} = N_{e2} =0$.  \\ \hline
			\end{tabular} 
			\caption{Firm Price Decisions} \label{tab:firmpricedecisions}
		\end{center}
	\end{table}
	\subsubsection{Two price case}
	
	In this case we assume that the firms each set two separate prices: one for EV spots ($c$) and one for the ICE drivers ($m$). The price for EVs can be thought of as the plug-in price on top of the price of parking (so $c = m + e$, where $e$ is this plug-in price). In principle $t$ could vary with the amount of electricity an EV consumes while parked (and could also vary with the underlying cost of electricity in a real-time market). Here, we do not model such considerations and simply view each EV as incurring the same additional cost $t$ (or equivalently we could view this as the average cost per EV). We also do not explicitly model the length of stay for each driver, so all drivers of the same class pay the same cost at each firm's location. Each firm in this stage seeks to choose prices to maximize revenue. Note that in this stage the cost of building parking spaces is sunk and so we can ignore this when optimizing revenue. Also, once the spaces are allocated the pricing decision for the two classes de-couple and thus a firm's optimal prices in response to the price of its opponent are given by: 
	\begin{equation*}
	c_i^{BR} = \underset{c_i}{\textrm{argmax}} \ c_iq_{ei}(N_{e1},N_{e2},c_1,c_2),
	\end{equation*}
	\begin{equation*}
	m_i^{BR} = \underset{m_i}{\textrm{argmax}} \ m_iq_{di}(N_{d1},N_{d2},m_1,m_2).
	\end{equation*}
	It can be shown that the objective in these optimization problems is concave in the prices and so the solution can be found by evaluating the first order optimality conditions which yields: 
	\begin{equation*}
	c_i^{BR}(N_{ej},c_j) = \frac{W_e(1+\alpha N_{ej}^{\epsilon}c_j)}{2(\alpha W_eN_{ej}^{\epsilon} + 1)},
	\end{equation*}
	\begin{equation*}
	m_i^{BR}(N_{dj},m_j) = \frac{W_d(1+\beta N_{dj}^{\epsilon}m_j)}{2(\beta W_dN_{dj}^{\epsilon} + 1)}. 
	\end{equation*}
	Hence, by solving for the fixed point of the best responses we get the following expressions for the unique equilibrium prices in the second stage:  
	\begin{align} \label{eq:optimalEVpricetwoprice}
	c_i^{*}(\mathbf{N_e}) = \frac{2\alpha W_e^2 N_{ei}^{\epsilon} + \alpha W_e^2 N_{ej}^{\epsilon} + 2W_e}{3\alpha^2W_e^2 N_{ei}^{\epsilon} N_{ej}^{\epsilon} + 4\alpha W_e(N_{ei}^{\epsilon} + N_{ej}^{\epsilon}) + 4},
	\end{align}
	\begin{equation} \label{eq:optimalICEpricetwoprice}
	m_i^{*}(\mathbf{N_d}) = \frac{2\beta W_d^2 N_{di}^{\epsilon} + \beta W_d^2 N_{dj}^{\epsilon} + 2W_d}{3\beta^2W_d^2 N_{ei}^{\epsilon} N_{dj}^{\epsilon} + 4\beta W_d(N_{di}^{\epsilon} + N_{dj}^{\epsilon}) + 4}. 
	\end{equation} 
	
	\subsubsection{Optimal single pricing case} 
	
	In this case we assume that each firm only sets one price that both EV drivers and ICE drivers pay to park (i.e., the plug-in price is always 0).  Each firm in this stage seeks to choose a single price $m$ to maximize revenue. Note again that in this stage the cost of building parking spaces is sunk and so we can ignore this when optimizing revenue. Thus a firm's optimal price in response to the price of its opponent are given by:
	\begin{equation*}
	m_i^{BR} =  \underset{m_i}{\textrm{argmax}} \ m_i[q_{ei}(\mathbf{N_e},\mathbf{m}) + q_{di}(\mathbf{N_d},\mathbf{m})]. 
	\end{equation*}
	\noindent Again this objective can  be shown to be concave in price, and therefore we can find the solution by examining the first order optimality conditions. Define 
	\begin{multline*}
	g^m_i(\mathbf{N_e},\mathbf{N_d},m_j)= W_e \left[\beta W_d ( 1 + \frac{N_{dj}}{N_{di}}) + \frac{1}{N_{di}^{\epsilon}} \right] ( 1 + \\ 
	\alpha N_{ej}^{\epsilon} m_j ) +  W_d\left[ \alpha W_e (1 + \frac{N_{ej}}{N_{ei}}) + \frac{1}{N_{ei}^{\epsilon}} \right] (1 +\beta N_{dj}^{\epsilon} m_j),
	\end{multline*}
	\begin{multline*}
	h^m_i(\mathbf{N_e},\mathbf{N_d},m_j) = (\alpha W_e N_{ej}^{\epsilon} + 1)(\beta W_d \left(1 + \right. \\
	\left. \frac{N_{dj}^{\epsilon}}{N_{di}^{\epsilon}} \right)  
	+ \frac{1}{N_{di}^{\epsilon}}) + (\beta W_d N_{dj}^{\epsilon} + 1)(\alpha W_e \left(1 + \frac{N_{ej}}{N_{ei}} \right) + \frac{1}{N_{ei}^{\epsilon}}). 
	\end{multline*}
	Using these, the best response price for each firm is:
	\begin{equation} \label{eq:priceoptimalsingleprice}
	m_i^{BR}(\mathbf{N_e},\mathbf{N_d},m_j) = \frac{g_i^m(\mathbf{N_e},\mathbf{N_d},m_j)}{2h_i^m(\mathbf{N_e},\mathbf{N_d},m_j)}.
	\end{equation} 
	As before we can solve for the fixed point of these best responses to get an expression for the equilibrium prices as a function of the number of parking spots for each class of driver ($m_i^{*}(\mathbf{N_e},\mathbf{N_d})$). The expression contains many terms and so is excluded from this write-up for clarity. 
	
	\subsubsection{``Naive" single pricing case} 
	
	In this case we also assume firms charge one price for parking. However firms are ``naive" about the entry of EVs into the market and so maintain the parking prices they had set before they had constructed EV chargers. Thus the best response function for the single price they charge is 
	\begin{equation*}
	m_i^{BR} = \underset{m_i}{\textrm{argmax}} \ m_iq_{di}(N_{d1} = \delta,N_{d2} = 1- \delta,m_1,m_2).
	\end{equation*}
	Note in this case the price is independent of any decision made in the first stage of the game as the objective function assumes $N_{d1} = \delta$, $N_{d2} = 1- \delta$. Again, we can solve this using the first-order optimality conditions to derive the best responses and then solve for the fixed point. Defining $N_i^{\epsilon} = \frac{N_i}{\epsilon}$, this yields 
	\begin{equation} \label{eq:pricenaivesingleprice}
	m_i^{*} = \frac{2\beta W_d^2 N_i^{\epsilon} + \beta W_d^2 N_j^{\epsilon} + 2W_d}{3\beta^2W_d^2 N_i^{\epsilon} N_j^{\epsilon} + 4\beta W_dN_i^{\epsilon} + 4\beta W_dN_j^{\epsilon} + 4}.
	\end{equation}
	
	\subsection{First stage: Capacity competition}
	
	Next, we turn to the first stage in which firms decide on how many EV spaces to construct accounting for the resulting price equilibrium in the second stage. We will look at two scenarios of how firms may respond to the government mandate in this stage. In the pricing stage, we assumed in some cases firms did not fully internalize their ability to charge EV drivers a different price than ICE drivers. We will do something similar in this stage. In one case we assume firms ``naively" follow the mandate while in the other they will more intelligently choose the number of spaces to optimize their objective function  in \eqref{eq:firmoptimization}. In the first scenario, the firms simply meet the mandate (i.e., adding $N_{ei} = rN_i$ as a constraint in \eqref{eq:firmoptimization}). In the second scenario, both firms optimize over the number of EV spots taking the mandate as a lower bound (adding \eqref{eq:governmentmandate} as a constraint in \eqref{eq:firmoptimization}). Table \ref{tab:firmcapacitydecisions} summarizes these two scenarios.
	
	\begin{table}[h!]
		\begin{center}
			\begin{tabular}{ |c|c| } 
				\hline
				Case & Capacity constraint  \\ \hline
				``Naive" mandate fulfillment & $N_{ei} = rN_i$  \\ \hline
				Optimal Capacity & \eqref{eq:governmentmandate} \\ \hline
			\end{tabular}
			\caption{Firm Capacity Decisions} \label{tab:firmcapacitydecisions}
		\end{center}
	\end{table}
	
	\begin{figure*}[t!]
		\centering	
		
		\begin{minipage}[t]{4cm}
			\begin{subfigure}[t]{1\textwidth}
				\centering	
				\includegraphics[scale=0.35]{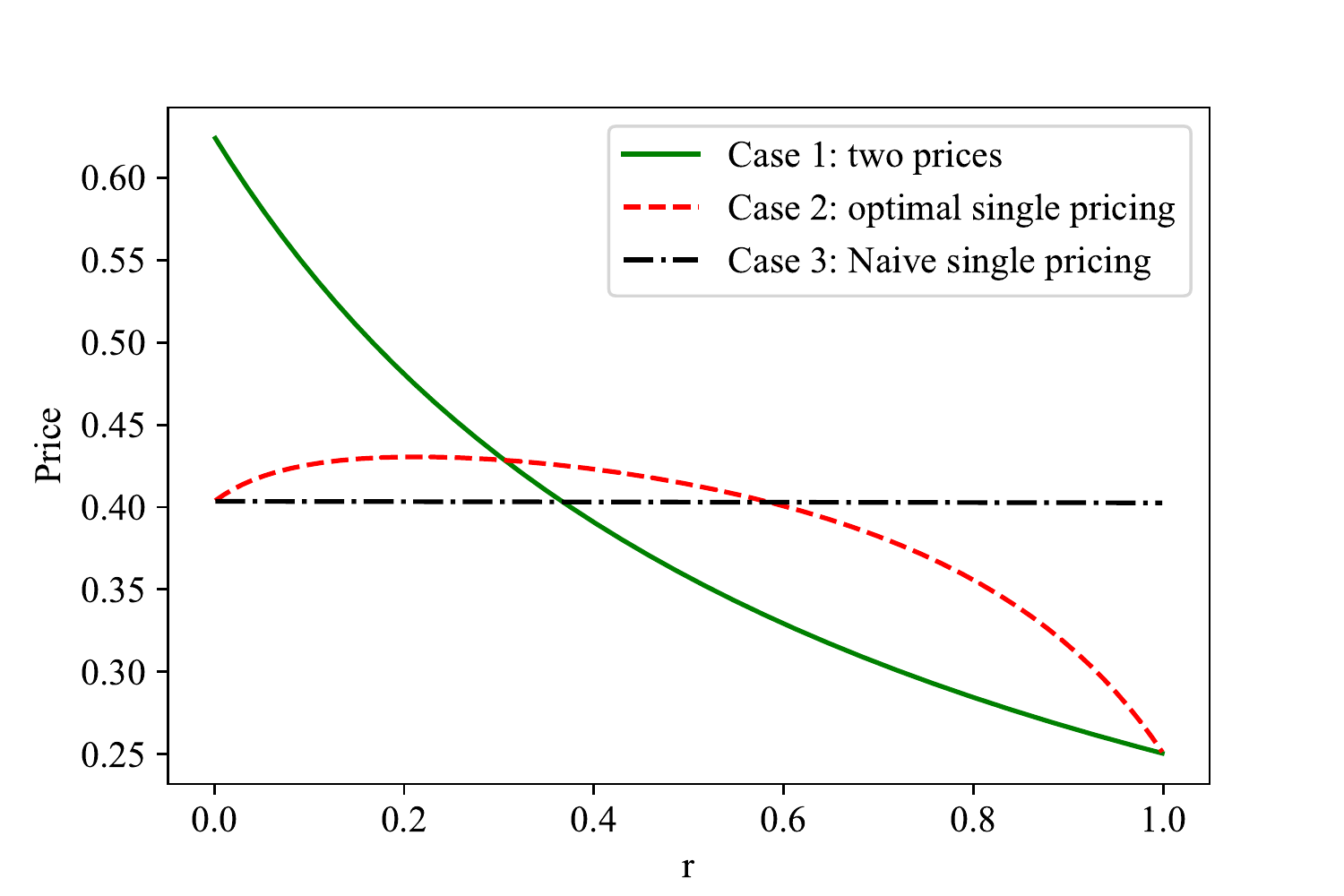} 
				\caption{Average EV price vs r} 
			\end{subfigure}
		\end{minipage}
		\hspace{1.1cm}
		\begin{minipage}[t]{4cm}
			\begin{subfigure}[t]{1\textwidth}
				\centering	
				\includegraphics[scale=0.35]{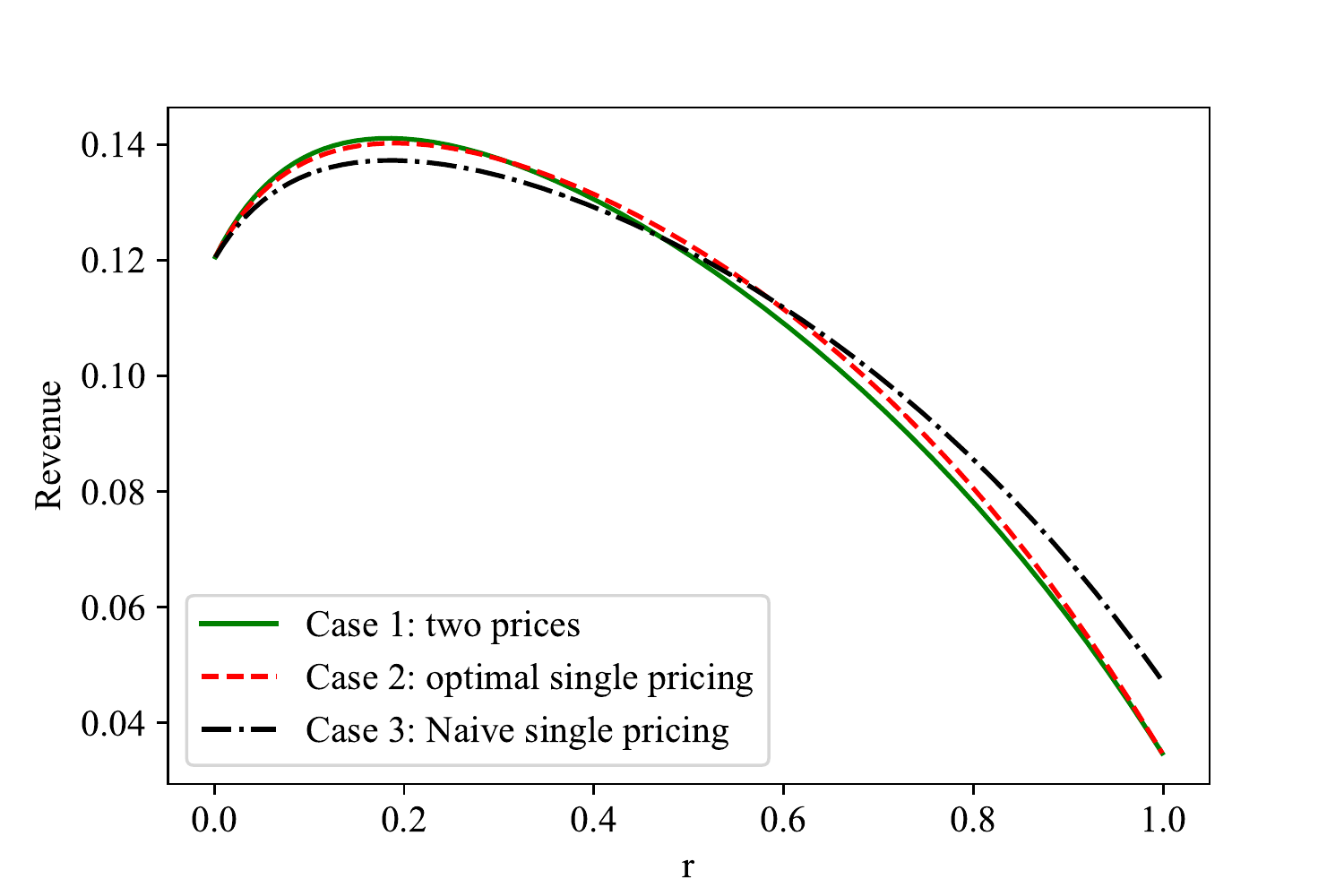} 
				\caption{Total Firm revenue vs r} 
			\end{subfigure}
		\end{minipage}
		\hspace{1.1cm}
		\begin{minipage}[t]{4cm}
			\begin{subfigure}[t]{1\textwidth}
				\centering	
				\includegraphics[scale=0.35]{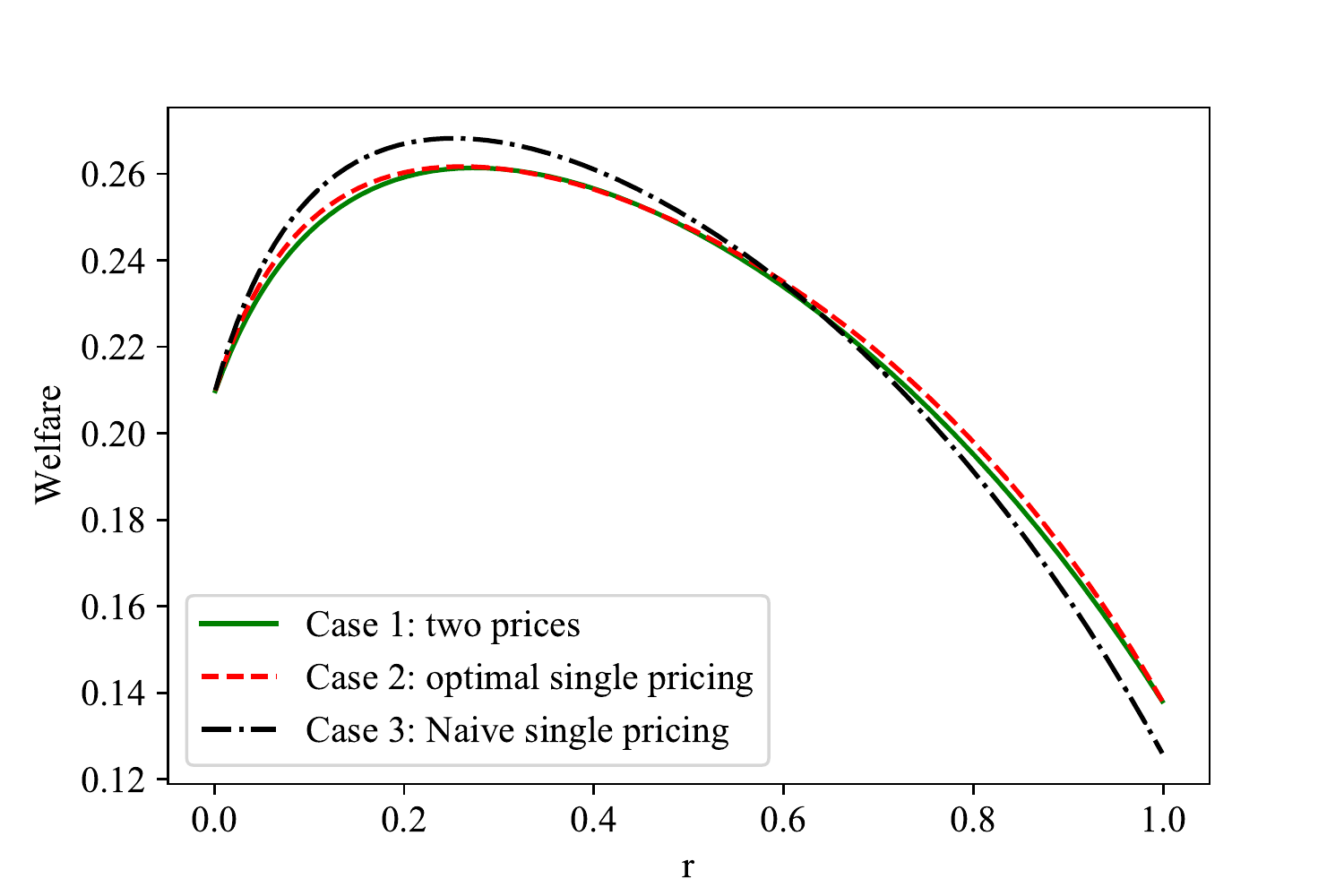} 
				\caption{Total Welfare vs r} 
			\end{subfigure}
		\end{minipage}
		
		\begin{minipage}[t]{4cm}
			\begin{subfigure}[t]{1\textwidth}
				\centering	
				\includegraphics[scale=0.35]{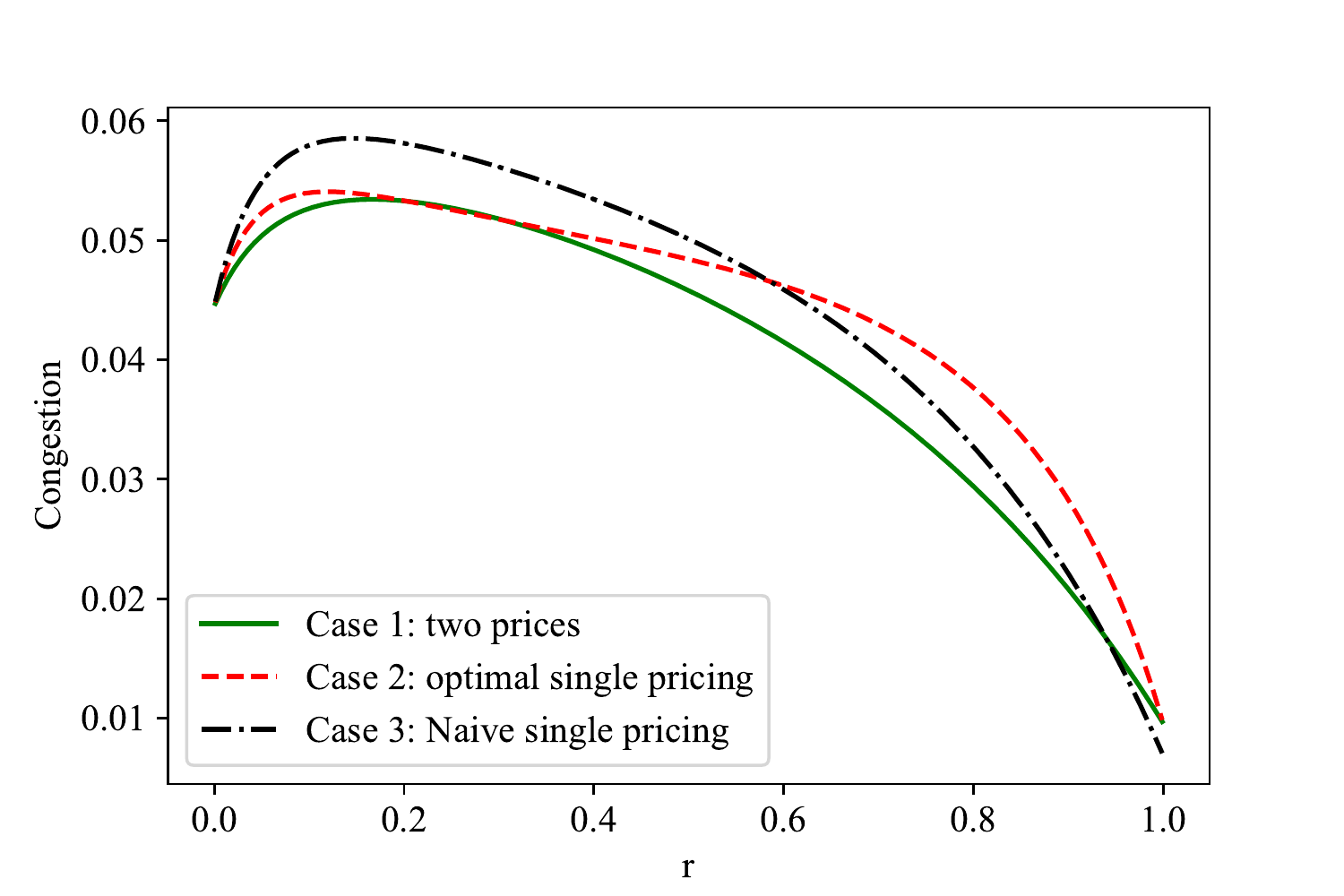} 
				\caption{Total Congestion vs r} 
			\end{subfigure}
		\end{minipage}
		\hspace{1.1cm}
		\begin{minipage}[t]{4cm}
			\begin{subfigure}[t]{1\textwidth}
				\centering	
				\includegraphics[scale=0.35]{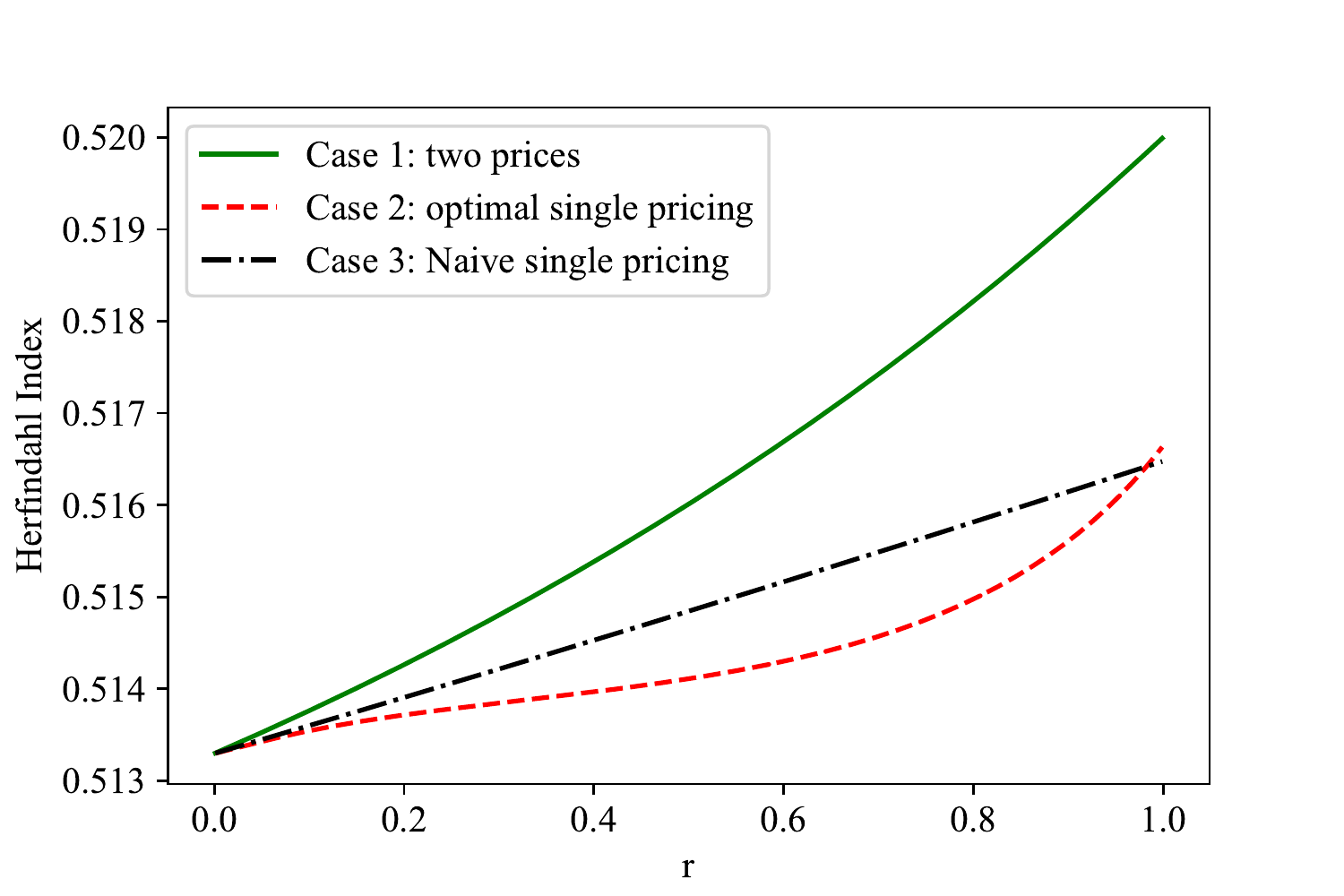}
				\caption{Herfindahl index ICE market vs r} 
			\end{subfigure}
		\end{minipage}
		\hspace{1.1cm}
		\begin{minipage}[t]{4cm}
			\begin{subfigure}[t]{1\textwidth}
				\centering	
				\includegraphics[scale=0.35]{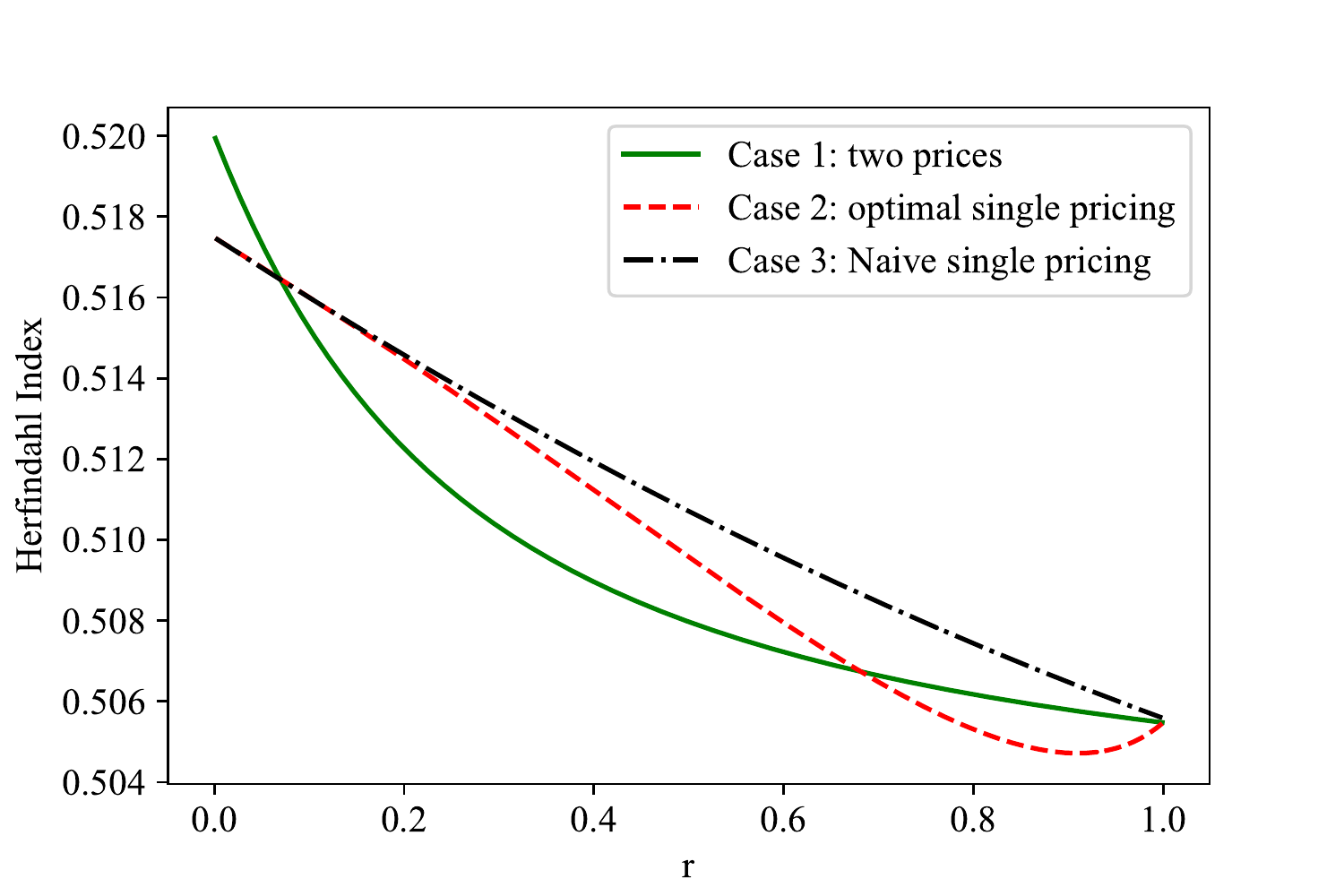} 
				\caption{Herfindahl index EV market vs r} 
			\end{subfigure}
		\end{minipage}
		
		\caption{Plots for $\alpha = 5$}
		
	\end{figure*}
	
	\section{Competitive Market: Results} \label{sec:competitiveresults}
	
	We now summarize the main results from analyzing the overall equilibrium under the different scenarios. \\
	
	\begin{prop} \label{prop:equilibriumexistence}
		For the ``naive" mandate fulfillment case, a sub-game perfect Nash equilibrium exists and is unique for all pricing schemes. In the optimal capacity case, a sub-game perfect Nash equilibrium exists if $\alpha W_e \leq 1$ , $\beta W_d \leq 1$ and $\epsilon = 1$ for the two price case and for any values in the ``naive" single price case. 
	\end{prop}
	
	This proposition shows that under certain scenarios a unique sub-game perfect Nash equilibrium exists. In general a closed form expression for this equilibrium appears to be difficult to obtain. In Section 5 we will instead numerically analyze properties of these equilibria. The uniqueness of the equilibria given in this theorem provides a guarantee that our results will not depend on the algorithm used to numerically find equilibria. The proof for this is provided in the appendix. 
	\begin{prop} \label{prop:maxquantitydriver}
		If the total number of spots across both firms for a class of driver is fixed, then the quantity of that class of driver serviced is maximized if when these spots are divided equally across the firms. 
	\end{prop}
	
	This makes intuitive sense. If there is a fixed number of spots for a given class of driver, the price will be lowest when firms have equal market power and thus compete the most. As a direct corollary  to this we have: 
	
	\begin{corollary} \label{cor:competitivewelfaremax}
		Consumer welfare is maximized when each firm has the same initial endowment ($\delta = 0.5$) and there is a unique equilibria. 
	\end{corollary}
	
	This follows from Proposition \ref{prop:maxquantitydriver}. When $\delta = 0.5$, it can be shown in all cases where equilibria are unique that each firm chooses the same number of spots for EVs and ICE drivers. This cannot happen if $\delta \neq 0.5$ as in this case the same number of spots cannot be constructed for both classes of drivers: at best this can only occur for one class. Again this result is intuitive, as consumer welfare is maximized when no firm has market power over another. 
	
	\begin{prop} \label{prop:bestresponseslopde}
		The best response prices in either market for the two price case is decreasing as the slope of demand in that market increases (i.e. as $\alpha$ or $\beta$ increases). The best response price in the ``naive" single price case increases as $\beta$ increases.
	\end{prop}
	
	\begin{proof}
		We prove this by showing that in the two price case the derivative of best response price for EV drivers is negative with respect to $\alpha$, $\forall \alpha \geq 0 $ and the derivative of best response price for ICE drivers is negative with respect to $\beta$, $\forall \beta \geq 0 $.
		\begin{multline*}
		\frac{\delta c^{*}_i}{\delta \alpha}(\mathbf{Ne}) = -(\frac{3\alpha W_e^3 Ne_i Ne_j}{\epsilon^2}) \times \\
		\frac{(\frac{2\alpha W_eN_{ei}}{\epsilon} +\frac{\beta W_eN_{ej}}{\epsilon} + 4)}{(3\alpha^2W_e^2 N_{ei}^{\epsilon} N_{ej}^{\epsilon} + 4\alpha W_e(N_{ei}^{\epsilon} + N_{ej}^{\epsilon}) + 4)^2} \leq 0, 
		\end{multline*}
		\begin{multline*}
		\frac{\delta m^{*}_i}{\delta \beta}(\mathbf{Nd}) = -(\frac{3\beta W_d^3 Nd_i Nd_j}{\epsilon^2}) \times \\
		\frac{(\frac{2\beta W_dN_{di}}{\epsilon} +\frac{\beta W_dN_{dj}}{\epsilon} + 4)}{(3\beta^2W_d^2 N_{di}^{\epsilon} N_{dj}^{\epsilon} + 4\beta W_d(N_{di}^{\epsilon} + N_{dj}^{\epsilon}) + 4)^2} \leq 0.
		\end{multline*}
		
		The ``Naive" pricing case's best response price function is equivalent to the ICE best price function of the two price case but with quantities of firm $i$ fixed as $N_i$ and firm $j$ fixed as $N_j$.
	\end{proof}
	This means that the smaller the market, the more each firm has an incentive to compete by lowering price. In the most extreme case where $\alpha = 0$ or $\beta = 0$, the firms do not compete at all (as there are an infinite number of consumers with the same valuation) and set the same price as a monopoly servicing that market. 
	
	\begin{figure*}[t!]
		\centering	
		\begin{minipage}[t]{4cm}
			\begin{subfigure}[t]{1\textwidth}
				\centering	
				\includegraphics[scale=0.35]{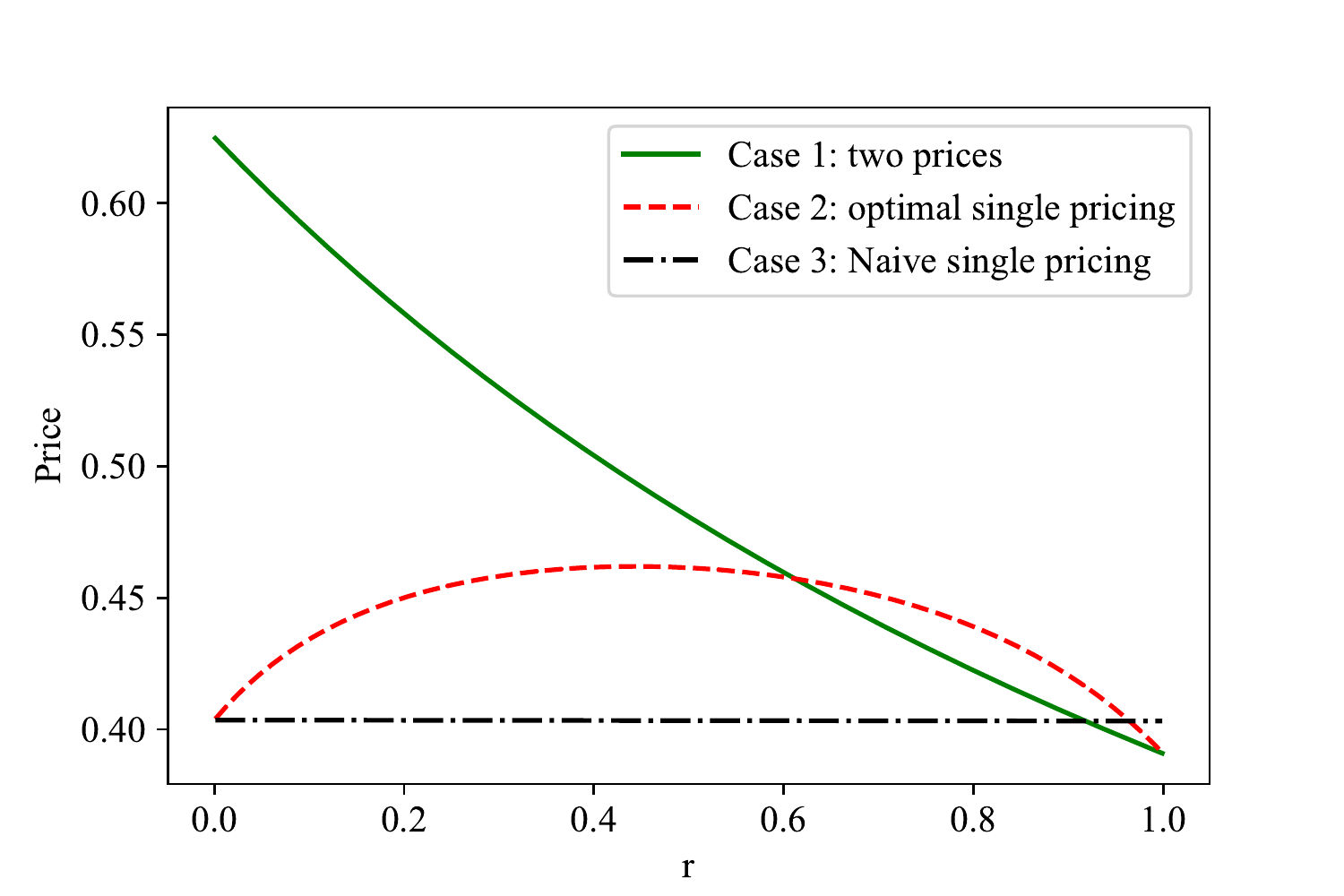} 
				\caption{Average EV price vs r} 
				\label{fig:evPrice5}
			\end{subfigure}
		\end{minipage}
		\hspace{1.1cm}
		\begin{minipage}[t]{4cm}
			\begin{subfigure}[t]{1\textwidth}
				\centering	
				\includegraphics[scale=0.35]{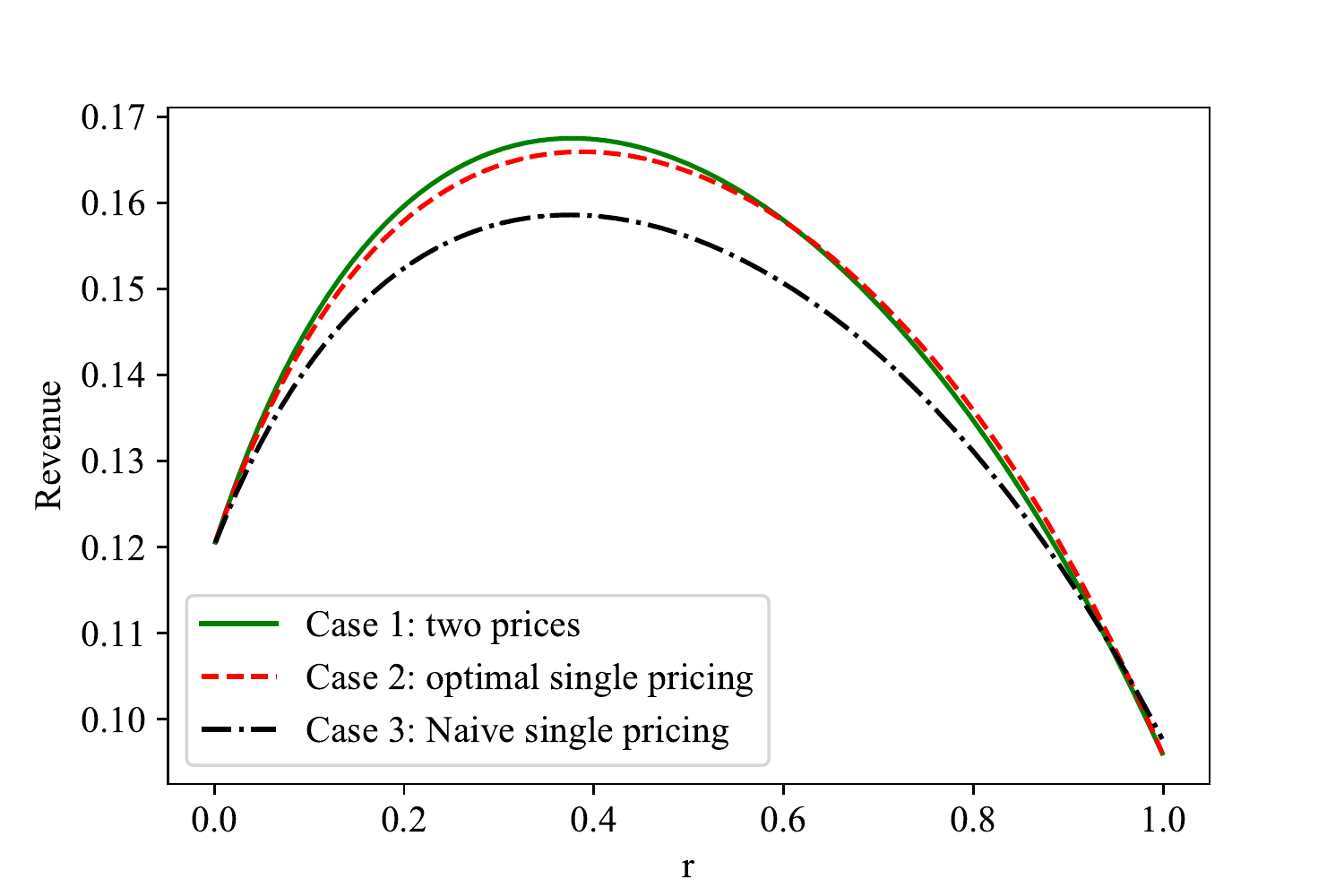} 
				\caption{Total Firm revenue vs r} 
			\end{subfigure}
		\end{minipage}
		\hspace{1.1cm}
		\begin{minipage}[t]{4cm}
			\begin{subfigure}[t]{1\textwidth}
				\centering	
				\includegraphics[scale=0.35]{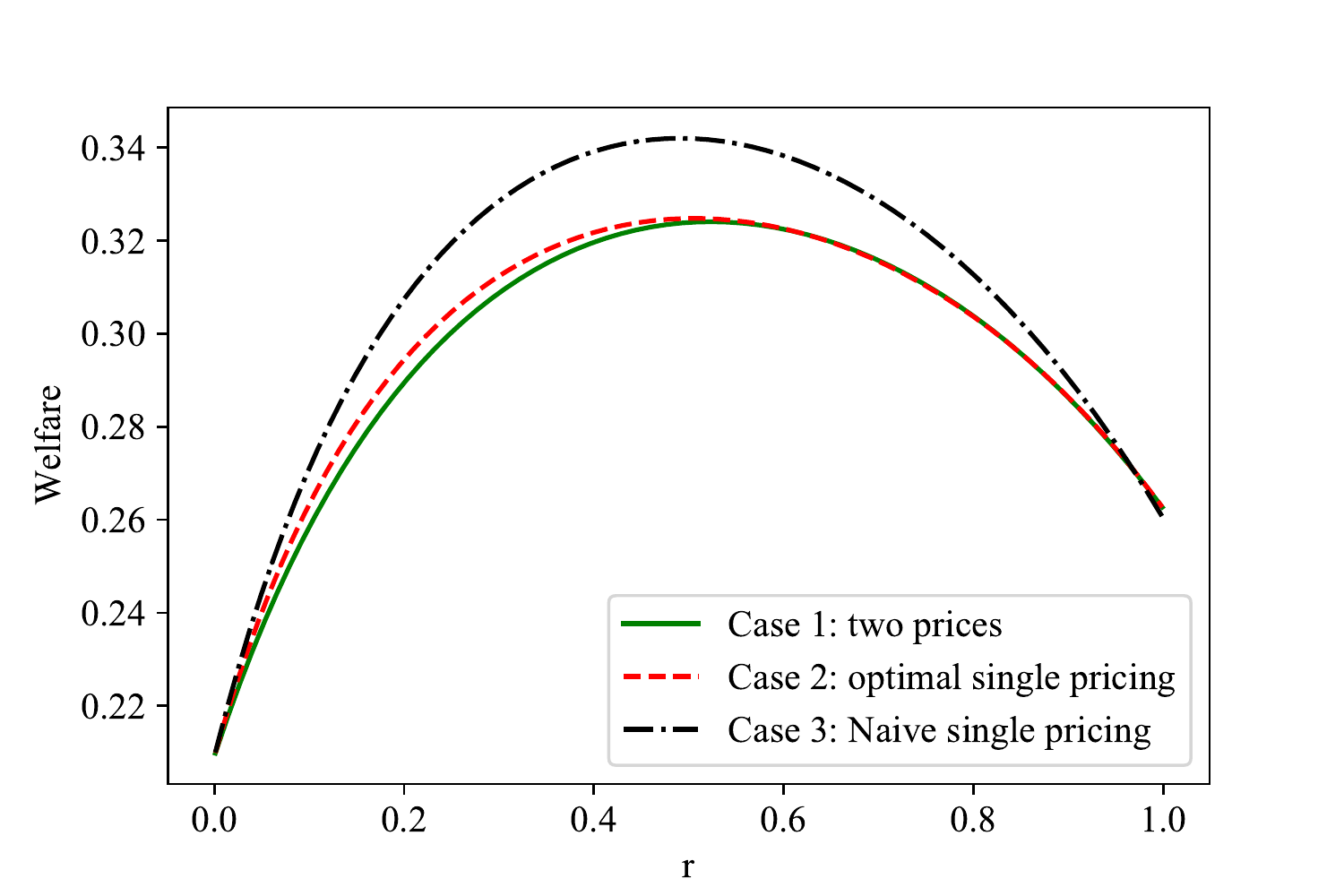} 
				\caption{Total Welfare vs r} 
			\end{subfigure}
		\end{minipage}
		
		\begin{minipage}[t]{4cm}
			\begin{subfigure}[t]{1\textwidth}
				\centering	
				\includegraphics[scale=0.35]{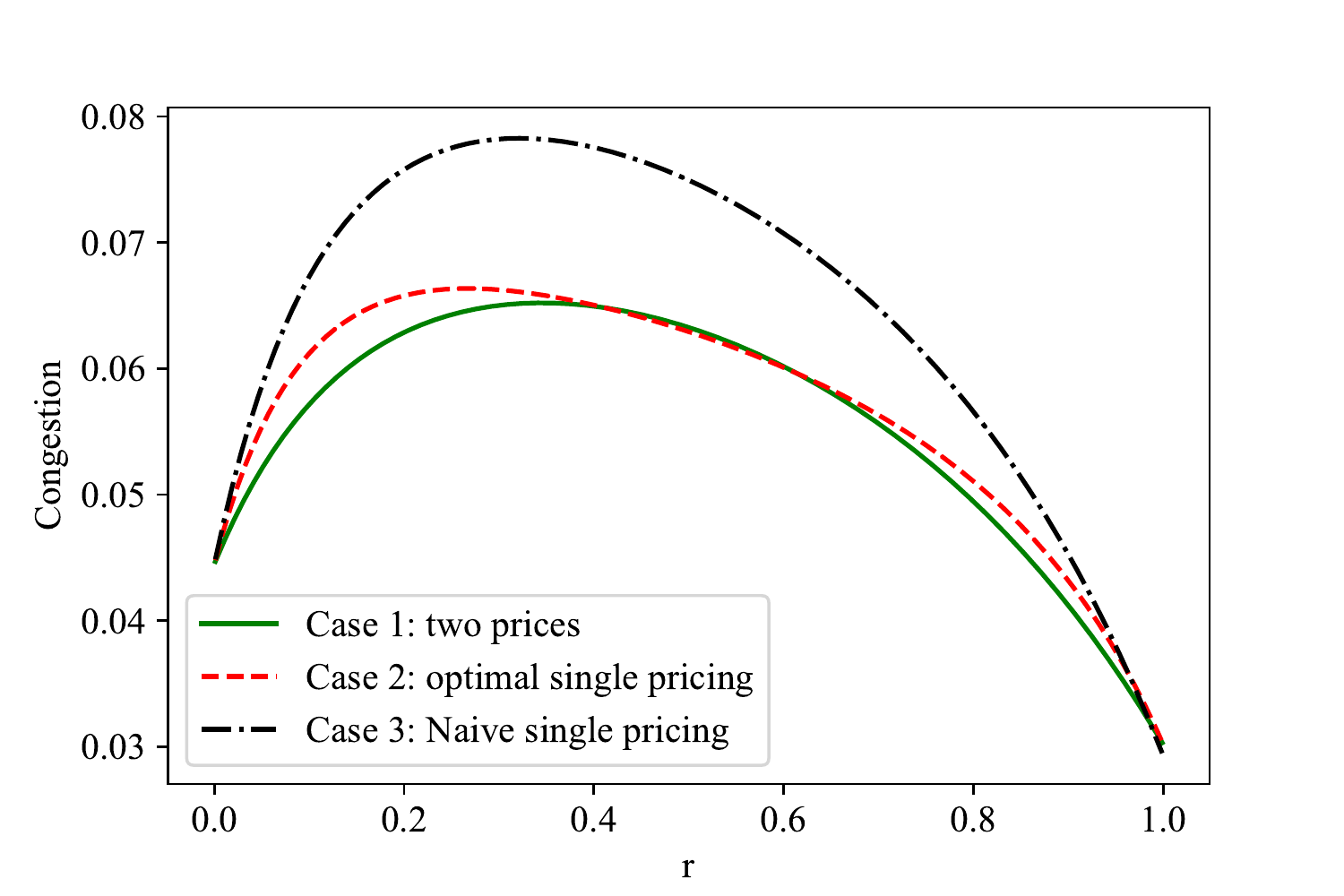}
				\caption{Total Congestion vs r} 
			\end{subfigure}
		\end{minipage}
		\hspace{1.1cm}
		\begin{minipage}[t]{4cm}
			\begin{subfigure}[t]{1\textwidth}
				\centering	
				\includegraphics[scale=0.35]{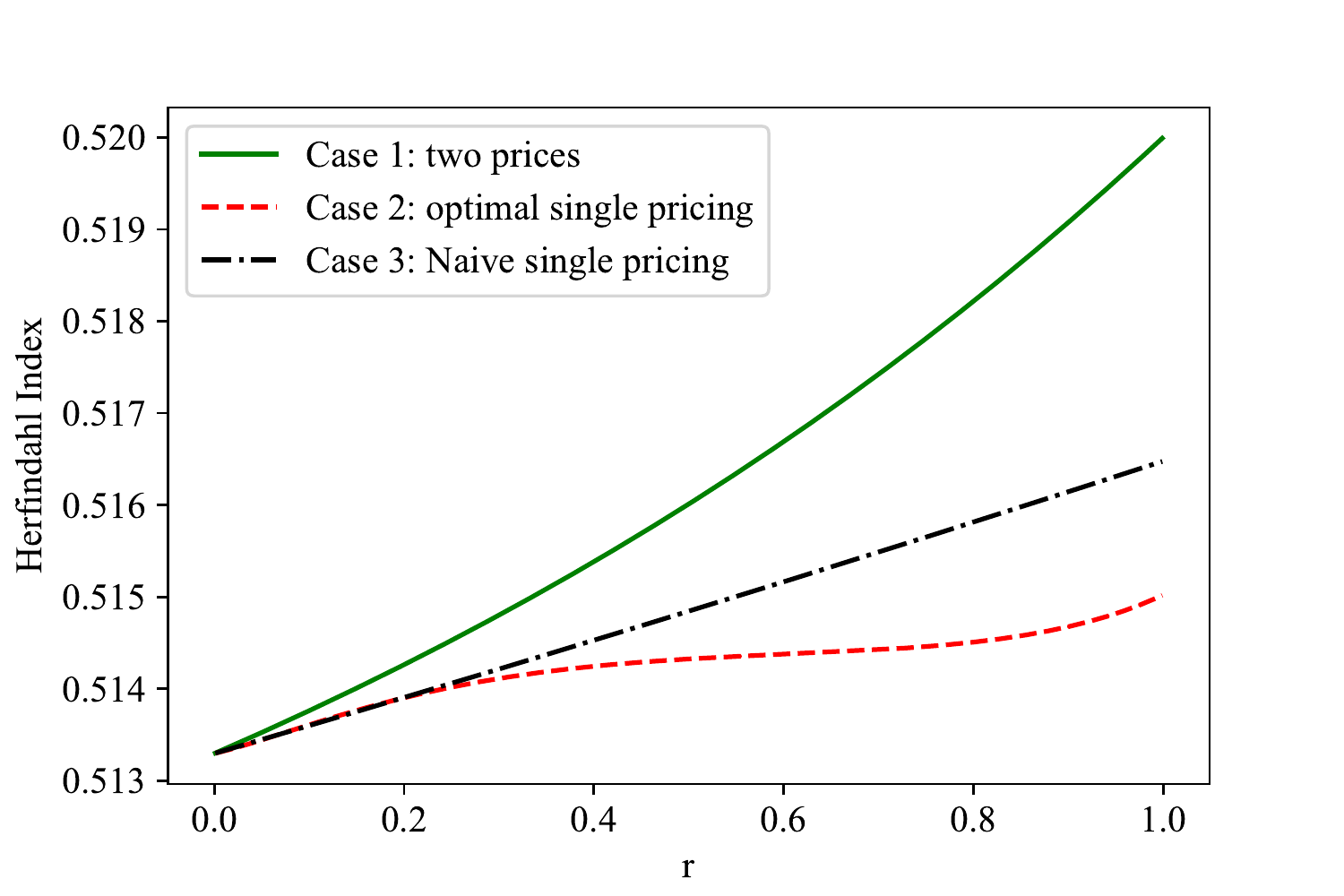}
				\caption{Herfindahl index ICE market vs r} 
			\end{subfigure}
		\end{minipage}
		\hspace{1.1cm}
		\begin{minipage}[t]{4cm}
			\begin{subfigure}[t]{1\textwidth}
				\centering	
				\includegraphics[scale=0.35]{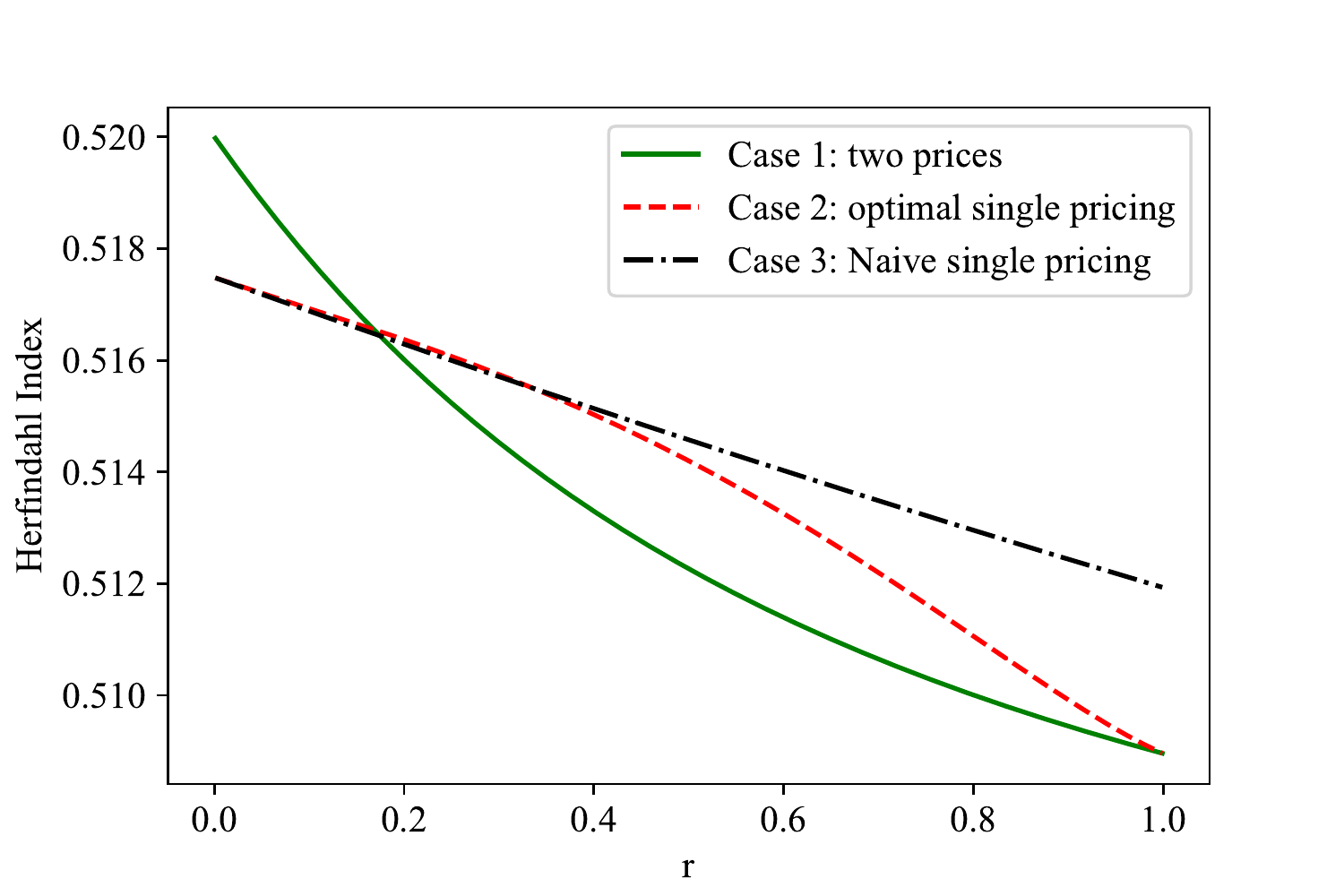} 
				\caption{Herfindahl index EV market vs r} 
			\end{subfigure}
		\end{minipage}
		
		\caption{Plots for $\alpha = 2$}
		
	\end{figure*}
	
	\section{Competitive Market: Numerical Case Study} \label{sec:competitivenumericalcasestudy}
	
	Next we present a numerical case study to see how differing levels of government mandate impact the market. We first define some values of interest.
	
	The first is the average price consumers of each class pay in equilibrium. For given prices $a_i$ the average price is defined as:
	\begin{equation*}
	a_{ave} = \sum_{i=1}^{2}\frac{q_{i}a_i}{q_{e1} + q_{e2}}.
	\end{equation*}	
	\noindent We are also interested in how competitive each market is. One way to measure this is the Herfindahl Index, a widely used measure of market concentration \cite{herfindahl1950concentration}. Given the quantity serviced $q_i$ at each firm $i$, this is defined as  
	\begin{equation*}
	H = \sum_{i=1}^{2}(\frac{q_{i}}{q_{1} + q_{2}})^2,
	\end{equation*}
	\noindent where larger values of $H$ indicate a less competitive market. 
	
	\subsection{``Naive" mandate fulfillment}
	In our first set of numerical results, we look at all three pricing approaches in the ``Naive" mandate fulfillment case. For all plots we set $W_d = 1$, $\beta = 1$, $W_e = 1.25$, $\alpha = 5$, $\epsilon = 1$, $p = 0$, $t=0$, and $\delta = 0.6$. In this case we have a large ICE market and a relatively small EV market, although the EV users have a larger utility for parking. We assume one of the garages is larger than the other (controlling 60\% of the parking spots). We then use the same values except set $\alpha = 2$. This can be viewed as the EV market expanding (the slope of demand decreases) while all other market conditions stay the same.
	
	We plot the average EV price, total firm profits, total welfare, total congestion (sum of congestion at all firms) and the Herfindahl index of both markets as a function of the government mandate ($r$). The plots for $\alpha = 5$ are shown in Figure 1 and the plots for $\alpha = 2$ are in Figure 2.  The values are plotted for all three pricing cases that were described above. We note some interesting features. First, the average EV price can be lower for the two price case than in any of the single price cases if the EV mandate is high enough when $\alpha = 5$ (the smaller EV market). However, when the EV market is larger ($\alpha = 2$), the average EV price in the two price case is never lower than in the naive single price case but can still be slightly lower than the optimal single price case when $r$ is large. Also, for $\alpha = 5$ the optimal single price can be lower than the naive price when $r$ is large but not for $\alpha = 2$. Intuitively, when the EV market is small and  a large number of spots are allocated to EV's the parking owners need to lower prices. We can also see that total firm profits can be larger in the naive pricing case than in any other case when $\alpha = 5$ and the EV mandate is large enough, though this is not the case for $\alpha = 2$. The intuition here is that although individually each firm would rather set two prices, fixing the price to that of the market before EVs enter reduces the ability of firms to compete with each other, which can increase aggregate profits. 
	
	Note in both cases the optimal total welfare occurs approximately when the mandate is equal to the size of the EV market ($\frac{1}{\alpha}$). At this value, naive single pricing maximizes welfare even though it generates lower firm profits. Apparently, the increased consumer surplus from lower prices makes up for the loss of profits, yielding improved welfare. Lastly, we note that the competitiveness of the regular market is lower in the two prices case than with the naive pricing when the EV mandate is set to the size of the EV market ($\frac{1}{\alpha}$) for both $\alpha =2$ and $\alpha = 5$. This is capturing the trade-off of inducing EV parking spots: if firms start competing on EVs, they may compete less over regular drivers and thus make the market less competitive for this class of consumers.
	
	\subsection{Optimal Capacity}
	
	We now examine the case where firms choose both $N_e$ and prices. We restrict ourselves to the ``two price" case where we have shown an equilibrium exists. We will assume the intrinsic cost of constructing spots is $t$, and the cost firms pay to construct them $p$ is a function of $t$ and government subsidies. This will impact aggregate welfare.

	We set $We = 1$, $\alpha = 1$, $W_d = 0.9$, $\beta = 0.33$, $\epsilon = 1$, $t = 0.1$, while varying $\delta$ between $0.5$ to $0.9$. We examine the effect of two government policies: one where they set a slightly lower mandate $r=0.33$ but fully subsidize construction, so $p=0$. We plot the welfare in Figure 3 as a function of $\delta$ with the four possible combinations of government policy as defined in Table \ref{tab:plotcaseesoptimalquantity}.

	\begin{table}[h!] 
		\begin{center}
			\begin{tabular}{ |c|c|c| } 
				\hline
				 & No Subsidy & Subsidy \\ \hline
				 No Mandate & (a) & (b)\\ \hline
				 Mandate & (c) & (d)\\ \hline
			\end{tabular} 
			\caption{Optimal capacity results} \label{tab:plotcaseesoptimalquantity}
		\end{center}
	\end{table}

	While aggregate welfare is decreased only slightly for all values of $\delta$ when any government action is taken, we see a large decrease in ICE consumer surplus (almost a $50\%$ drop-off) in all cases of government actions, but with a large increase in EV welfare (up to an almost $400\%$ increase) for all values of $\delta$. Effectively the government is only somewhat inefficiently trading off ICE welfare for EV welfare by either forcing or inducing (through subsidies) construction of EV chargers. Having both a large subsidy and relatively large mandate does not have much more impact beyond using just one policy given: they both produce similar effects independently and together. However, subsidies lead to government welfare loss and increases both firm profits at all levels of $\delta$, while mandates lead to lower firm profits and no cost to the government. 
	
	\begin{figure*}[t!] \label{fig:optimalcapacity}
		\centering	

		\begin{minipage}[t]{3.2cm}
			\begin{subfigure}[t]{1\textwidth}
				\centering	
				\includegraphics[scale=0.28]{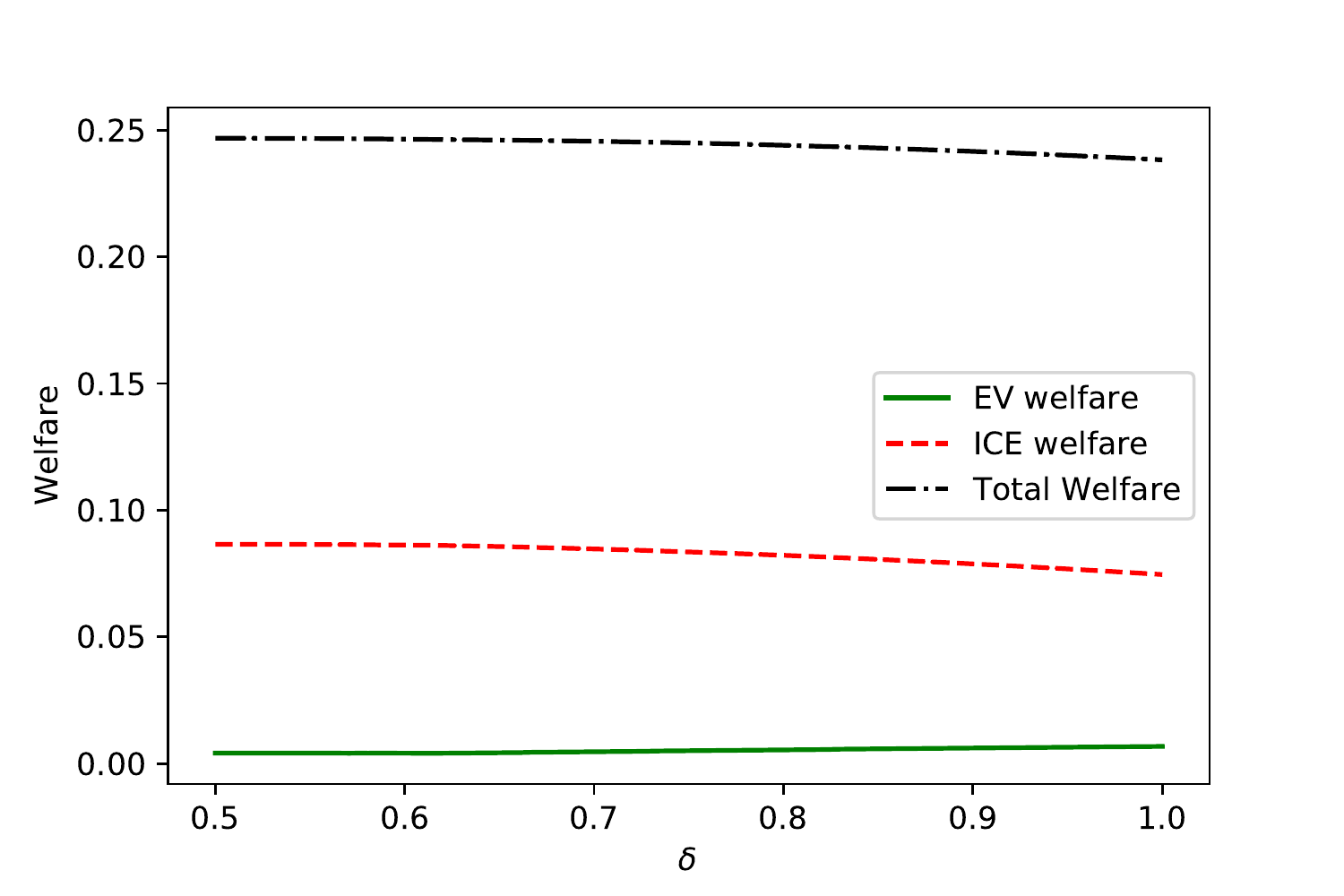} 
				\caption{Welfare, no government action} 
				\label{fig:welfarenogovt}
			\end{subfigure}
		\end{minipage}
		\hspace{0.9cm}
		\begin{minipage}[t]{3.2cm}
			\begin{subfigure}[t]{1\textwidth}
				\centering	
				\includegraphics[scale=0.28]{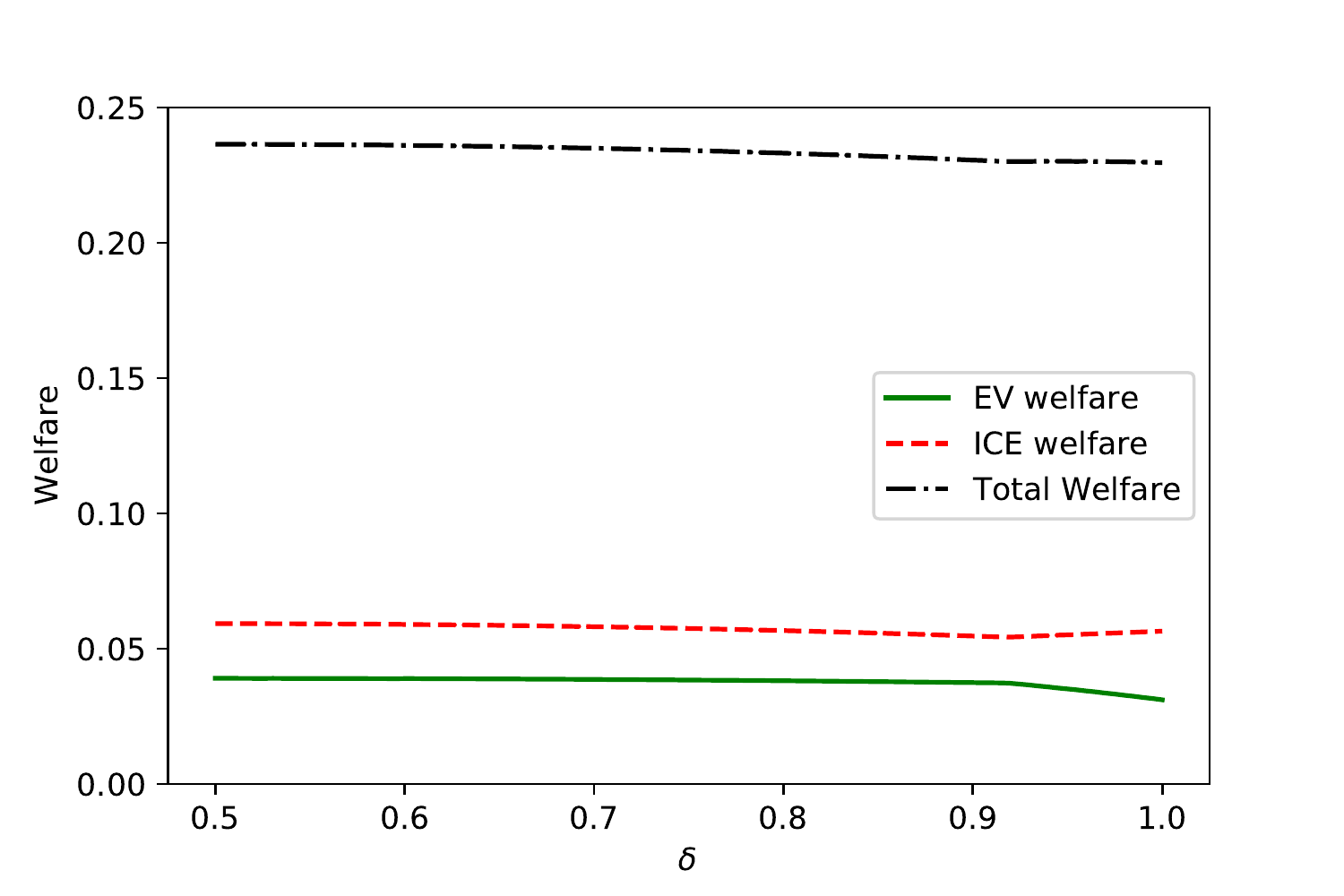} 
				\caption{Welfare, full subsidy} 
				\label{fig:welfaresubsidy}
			\end{subfigure}
		\end{minipage}
		\hspace{0.9cm}
		\begin{minipage}[t]{3.2cm}
			\begin{subfigure}[t]{1\textwidth}
				\centering	
				\includegraphics[scale=0.28]{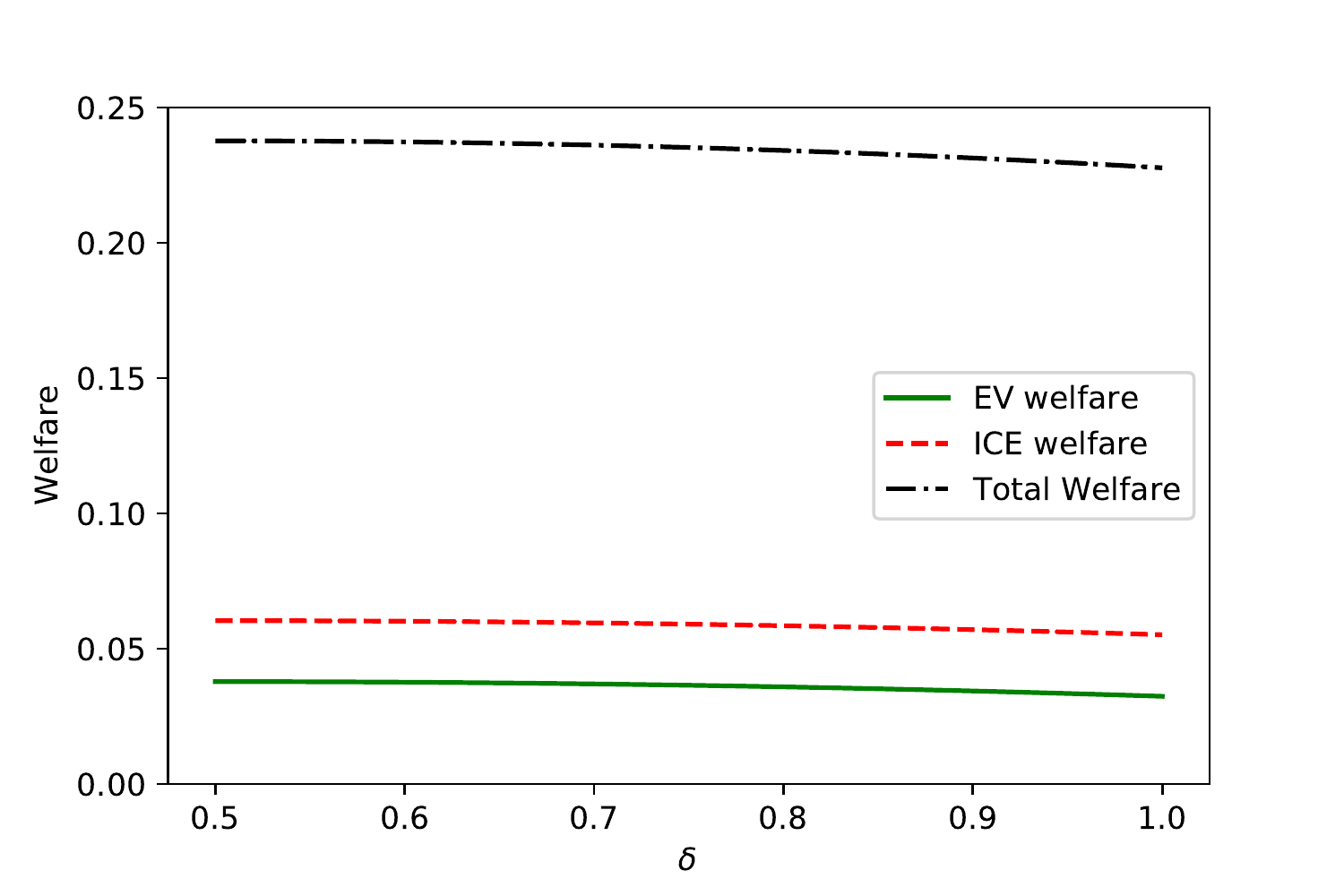} 
				\caption{Welfare, government mandate} 
				\label{fig:welfaremandate}
			\end{subfigure}
		\end{minipage}
		\hspace{0.8cm}
		\begin{minipage}[t]{3.2cm}
			\begin{subfigure}[t]{1\textwidth}
				\centering	
				\includegraphics[scale=0.28]{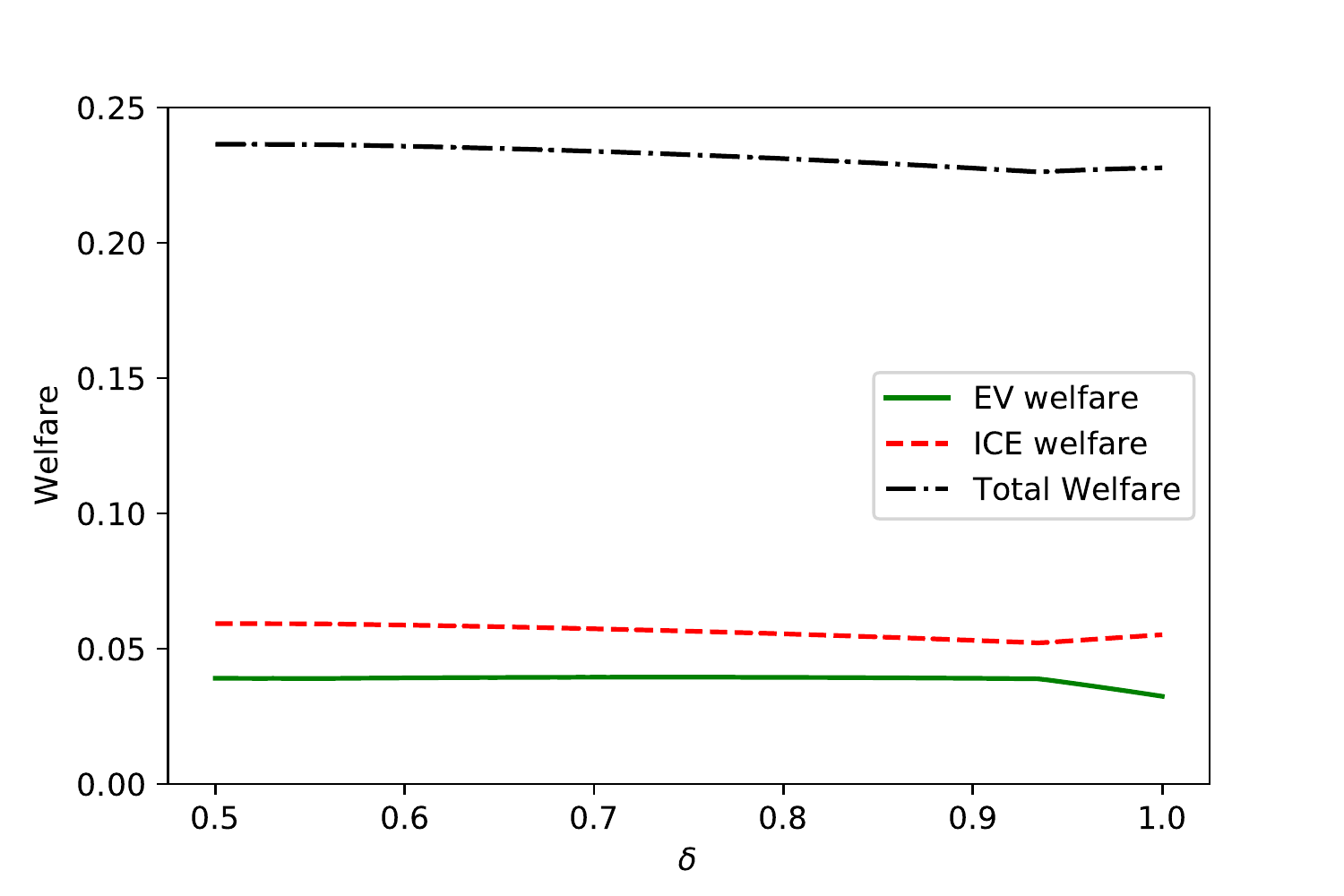} 
				\caption{Welfare, mandate and subsidy} 
				\label{fig:welfaremandateandsubsidy}
			\end{subfigure}
		\end{minipage}
	
		\caption{Plots for Optimal Capacity} 
		
	\end{figure*}

	\section{Monopolist with Stochastic Demand: Model} \label{sec:monopolistmodel}
	
	\noindent We now turn to our second model in which there is uncertainty in the demand for EV drivers. We consider a similar model as in Section II, with two changes. First, we consider a monopolist instead of a competitive market. Second, the demand function of both types of drivers will change. 
	
	\subsection{Consumers}
	
	We will assume there is a mass of regular drivers all with the same valuation for parking, $W_d$. This mass is sufficiently large (essentially unbounded) such that at any positive price, the monopolist can fill up its parking spots until the marginal value is zero. For the EV market, we similarly assume all drivers have the same valuation for parking, $W_e > W_d$, but there are only $q_{max}$ EV consumers, where $q_{max}$ depends on the state of the world.  Specifically, we define a probability distribution $\mathbf{\pi} = \{\pi_1,\pi_2,\ldots,\pi_n\}$ over a set $\mathbf{q} = \{q_1, q_2,\ldots,q_n\}$, so that $q_{max} = q_i$ with probability $\pi_i$. The monopolist sets both the capacity ($N_e$) and prices for ICE and EV drivers ($m$ and $c$, respectively) before this distribution is realized. We can therefore write the marginal utility of EV drivers (as a function of $q_{max}$) and ICE drivers as 
	
	\begin{equation*} \label{eq:utilityICEmonopolist}
	U_d(q_d,N_e,m) = W_d -  \epsilon \frac{q_d}{1 - N_e}- m,
	\end{equation*}
	\begin{align*} \label{eq:utilityEVmonopolist}
	U_e(q_e,N_e,c,q_{max}) &= W_e - \epsilon \frac{q_e}{N_e} - c. \\
	\text{s.t.} \ \ \ & q_e \leq q_{max}
	\end{align*}
	
	\subsection{Monopolist}
	
	We assume there exists a monopolist who already owns a mass of parking spots suitable for ICE drivers. Like in the previous model, we normalize this number of spots to be 1. This monopolist cannot construct new parking spots, but is able to convert a portion of their spots to include an EV charger at cost $p$. This means choosing a number of EV spots $ 0 \leq N_e \leq 1$, leaving $1 - N_e$ spots for ICE drivers.  The monopolist then chooses a price $c$ for EV drivers and $m$ for ICE drivers (i.e., we are only examining the  ``Two Price'' case as described in the competitive model). They do this before realizing what state of the world they are in (i.e., before the size of the EV market is realized).
	
	Assuming $m \leq W_d$ and $c \leq W_e$, the demand for ICE drivers and EV drivers (as a function of the realization $q_{max}$), respectively, is 
	\begin{equation*}
	D_d(N_e,m) = (1 - N_e)\frac{W_d - m}{\epsilon},
	\end{equation*}
	\begin{equation} \label{eq:EVdemandstochastic}
	D_e(N_e,c,q_{max}) = \min \left\{ N_e \frac{W_e - c}{\epsilon},q_{max} \right\}.
	\end{equation}
	The monopolist is risk neutral and wants to maximize their expected profit from the ICE market and EV market, which will we define, respectively, as
	
	\begin{equation} \label{eq:ICErevenuemonopolist}
	\Pi^d(m,N_e) =  mD_d(N_e,m),
	\end{equation}
	\begin{equation} \label{eq:EVrevenuemonopolist}
	E[\Pi^e(c,N_e)] = \sum_{i = 1}^{i = n} c\pi_i  D_{e}(N_e,c,q_i).
	\end{equation}
	
	Combining \eqref{eq:ICErevenuemonopolist} and \eqref{eq:EVrevenuemonopolist} as well as the cost of constructing spots, the overall profit maximization problem of the monopolist is:
	\begin{equation} \label{eq:monopolistoptimization}
	\begin{aligned}
	& \underset{N_e,m,c}{\text{maximize}}
	& &  \Pi_d(m,N_e) + E[\Pi_e(c,N_e)] - pN_{e} \\
	& \text{subject to}
	& & 0 \leq N_e \leq 1.
	\end{aligned}
	\end{equation}
	Similar to the competitive model, we analyze this problem by viewing the monopolist's decision in two stages: first choosing $N_e$ and then choosing prices $m$ and $c$. 
	
	\section{Monopolist with Stochastic Demand: Pricing} \label{sec:monopolistpricing}
	
	In this section, we analyze the prices the monopolist chooses as a function of $N_e$. We note that once the proportion of EV spots is chosen, the decisions to price EV and ICE drivers are independent as they cause no congestion to each other and the cost of the spots is sunk, and so we examine each separately. 
	
	\subsection{ICE driver Pricing} 
	
	We can find the optimal ICE price by solving
	
	\begin{equation*} 
	\begin{aligned}
	& \underset{m}{\text{maximize}}
	& & \Pi^d(m) = m(1-N_e)\frac{W_d - m}{\epsilon}.  \\
	\end{aligned}
	\end{equation*}
	
	This problem is concave in $m$ and so we can find the solution by examining the first order optimality conditions which gives 
	
	\begin{equation} \label{eq:optmonopolistICEprice}
	m^{*}(N_e) = \frac{W_d}{2}.
	\end{equation}
	
	We can plug \eqref{eq:optmonopolistICEprice} back into the optimization problem to get an expression for ICE revenue as a function of $N_e$  
	
	\begin{equation} \label{eq:ICErevmonopolist}
	\Pi^d(N_e) = \frac{(1 - N_e)W_d^2}{4\epsilon}.
	\end{equation}
	
	\noindent We can see this is simply a linear function in $N_e$, which increases in the ICE driver's utility ($W_d$) and decreases in $\epsilon$. 
	
	\subsection{EV driver Pricing} 
	
	Note that fixing $N_e$, we can break down the EV quantities serviced in each realization into three broad cases based on the EV price ($c$) the monopolist sets. We will define each case by what quantity of EV drivers $q_e$ they service for a given $q_{max} \ \in \mathbf{q}$. 
	
	The first two cases are uniquely defined by a target realization $q_t \ \in \  \mathbf{q}$. Case 1 occurs when the monopolist sets a price such that $q_{e} = q_{max}$ in realizations where $q_{max} \leq q_t$ while $q_{e} = q_t$ when $q_{max} > q_t$. In other words the monopolist sets an EV price such that they exactly service the target quantity when there is a large enough demand for it, and fully service all realizations smaller than the target. Case 2 occurs when the monopolist sets a price such that $q_{e} = q_{max}$ when $q_{max} \leq q_t$ and $q_{e} = q_{r}$ for some $q_t < q_{r} < q_{t+1}$. This means the monopolist serves exactly $q_{max}$ for any realization where  $q_{max} \leq q_t$, but serves some amount between the target realization and the next highest realization in all states of the world such that $q_{max} > q_t$.
	
	Finally the monopolist can set a price such that $q_e < q_1$ in all realizations of $q_{max}$, which we can call Case 3. These cases are summarized in Table \ref{tab:EVpricingcasesgeneraldist} .
	
	\begin{table}[h!] 
		\begin{center}
			\begin{tabular}{ |c|c| } 
				\hline
				Case & EV quantity serviced \\ \hline
				Case 1 & $q_{e} = q_i$ if $q_{max} \leq q_t$, else $q_{e} = q_t$  \\ \hline
				Case 2 & $q_{e} = q_i$ if $q_{max} \leq q_t$, else $q_{e} = q_r$ ($ q_t < q_r < q_{t+1}$) \\ \hline
				Case 3 & $q_{e} < q_1 \ \forall q_{max} $ \\ \hline
			\end{tabular} 
			\caption{Monopolist quantity strategies.} \label{tab:EVpricingcasesgeneraldist}
		\end{center}
	\end{table}
	For a given number of possible realization $n$ the possible cases can be ordered as follows (which is increasing in the size of the maximum number of EVs served over all realizations). \\ 
	
	[Case 3, Case 1 ($t=1$), Case 2 ($t=1$), Case 1 ($t=2$), ... , Case 2 ($t=n-1$), Case 1 ($t=n$)]. \\
	
	This allows us to describe any feasible set of quantities the monopolist serves with $n$ target quantities of Case 1 ($t = [1,2,...,n]$, i.e., the monopolist can seek to at most service any given realization) and $n - 1$ target quantities for Case 2 ($t = [1,2,...,n-1]$ as the monopolist cannot service more than the largest realization). Given the threshold structure of EV demand in \eqref{eq:EVdemandstochastic}, note  that for any realization where $q_{e} \leq q_{max}$, the monopolist will serve the same $q_{e}$ in any realization with a larger quantity of EV drivers. 
	
	We now analyze the pricing in each of these cases (with a given target realization $q_t$ where applicable) for a fixed $N_e$. \\
	
	\subsubsection*{Case 1: $q_{e} = q_i$ if $q_{max} \leq q_t$, else $q_{e} = q_t$} 
	
	In this case, for a given target quantity $q_t$, the monopolist sets the price such that they get exactly $q_t$ in any realization where the number of arrivals is greater than or equal to $q_t$. By setting \eqref{eq:EVdemandstochastic} equal to $q_t$ and solving for $c$, we find the price is
	
	\begin{equation} \label{eq:pricecase1}
	c_1^{*}(N_e,t) = W_e - \epsilon\frac{q_t}{N_e}.
	\end{equation}
	
	\subsubsection*{Case 2: $q_{e} = q_i$ if $q_{max} \leq q_t$, else $q_{e} = q_r$ ($ q_t < q_r < q_{t+1}$)} 
	
	In this case the monopolist wants to set the optimal price which is low enough to get a quantity $>q_t$ while large enough to get some amount $<q_{k+1}$. This can be expressed as the following optimization problem 
	
	\begin{equation} \label{eq:case2twostageoptimization}
	\begin{aligned}
	& \underset{c}{\text{max}}
	& & \Pi^{e}_3(c) = c(\sum_{i = 1}^{i = t}\pi_i q_i + \sum_{j = t+1}^{j = n}\pi_j N_e\frac{W_e - c}{\epsilon})  \\
	& \text{subject to}
	& & W_e - \epsilon \frac{q_t}{N_e} < c < W_e - \epsilon \frac{q_{t+1}}{N_e}.
	\end{aligned}
	\end{equation}
	
	Again we can make the constraint set closed by relaxing it to be $W_e - \epsilon \frac{q_t}{N_e} \leq c \leq W_e - \epsilon \frac{q_{k+1}}{N_e}$. If any constraint is tight, then this case is not optimal for the given setting. If the constraints are not tight, we can find the optimal price by examining the first order optimality conditions , which gives us
	
	\begin{equation} \label{eq:pricecase2}
	c_2^{*}(N_e,t) = \frac{W_e}{2} + \frac{\epsilon \sum_{i=1}^{t}\pi_iq_i}{2\sum_{j=t+1}^{n}\pi_j N_e}.
	\end{equation}
	
	This price can be viewed as the price if there was no constraint on the size of the EV market ($\frac{W_e}{2}$) (see Case 3)  which is decreasing in the target $t$. \\
	
	\subsubsection*{Case 3: $q_{e} < q_1 \ \forall q_{max} $} 
	
	In this case, the monopolist sets the optimal price as a function of $N_e$ such that they get $q_e < q_1$ in any realization. For this case to be feasible, the monopolist has to set a high enough price such that $D_e(N_e,c,q_i) < q_1$. Hence the pricing scheme in this case is equivalent to the following optimization problem
	
	\begin{equation} \label{eq:case3optimization}
	\begin{aligned}
	& \underset{c}{\text{maximize}}
	& & \Pi^{e}_1(c) = cN_e\frac{W_e - c}{\epsilon} \\
	& \text{subject to}
	& & c > W_e - \epsilon \frac{q_1}{N_e}.
	\end{aligned}
	\end{equation}
	
	Again we can relax the constraint to be $c \geq W_e - \epsilon \frac{q_1}{N_e}$, and note that if this constraint is ever tight we are not in this case. If this occurs we will say that this case is no longer feasible, as case 2 is strictly superior. When the constraint is not tight, we find the optimal price is
	\begin{equation} \label{eq:optmonopolistEVpricecasea}
	c_3^{*}(N_e,t) = \frac{W_e}{2}.
	\end{equation}
	This is the same pricing scheme as for the ICE drivers (a constant price independent of $N_e$). This makes sense as assuming ex ante the monopolist sets a price such that they serves $q_e < q_1$ EV drivers is the same as being unconstrained by the size of the EV market. 
		
	\section{Monopolist with Stochastic Demand: Capacity} \label{sec:stochasticarrivalcapacity}	
		
	Next, we turn to the Monopolist's optimal choice of $N_e$.  Using the results of the previous section, we can now formulate this as 
	\begin{align} \label{eq:firststagemonopolist}
	\underset{N_e}{\text{maximize}} & \ \ m^{*}(N_e)D_d(N_e,m^{*})\  + \\
	& c^{*}(N_e,t)\sum_{r = 1}^{r = n} \pi_r  D_{e}(N_e,c^{*}(N_e,t),q_r) - pN_{e} \nonumber \\ 
	\text{subject to} & \ \ 0 \leq N_e \leq 1. \nonumber
	\end{align}
	At each value of $N_e$, the optimal EV price $c^{*}(N_e,k)$ is defined by whichever pricing case and target quantity pair leads to the highest profit (at a given $N_e$, any given case could lead to this). The optimal ICE price $m^{*}(N_e)$ is defined by \eqref{eq:optmonopolistICEprice} regardless of case (and in fact is also independent of $N_e$). The overall optimal profit of the monopolist is defined by the best choice of EV spots $N_e^*$ that solves this problem. This choice also solves \eqref{eq:monopolistoptimization} (where the optimal $m$ and $c$ are determined by this choice $N_e^*$). We now examine an important property of $N_e^{*}$.
	 	
	\begin{theorem} \label{thm:stochasticdemand}
		Assume $\frac{W_e}{2} > \epsilon q_1$ and $\frac{4}{\epsilon}(W_e^2 - W_d^2)  > p$. Any optimal pricing strategy for the monopolist is to set a price $c$ such that for some $t$, $1 \leq t \leq n$, the monopolist services exactly $q_t$ for every realization of the distribution such that $q_e \geq q_t$, and exactly $q_i$ for any realization where the quantity of EV drivers is $q_i < q_t$.
	\end{theorem}

	 In other words, this shows that under the assumptions in this theorem the optimal price is always in Case 1.  Hence, to determine the optimal $N_e$, we only need to consider the $n$ prices corresponding to Case 1. We can prove this by examining the structure of \eqref{eq:firststagemonopolist} given how $c^{*}(N_e)$ is defined in each case. The proof for this is found in the appendix.

	\begin{figure*}[t!] \label{fig:stochasticdemand}
		\centering	
		
		\begin{minipage}[t]{4cm}
			\begin{subfigure}[t]{1\textwidth}
				\centering	
				\includegraphics[scale=0.35]{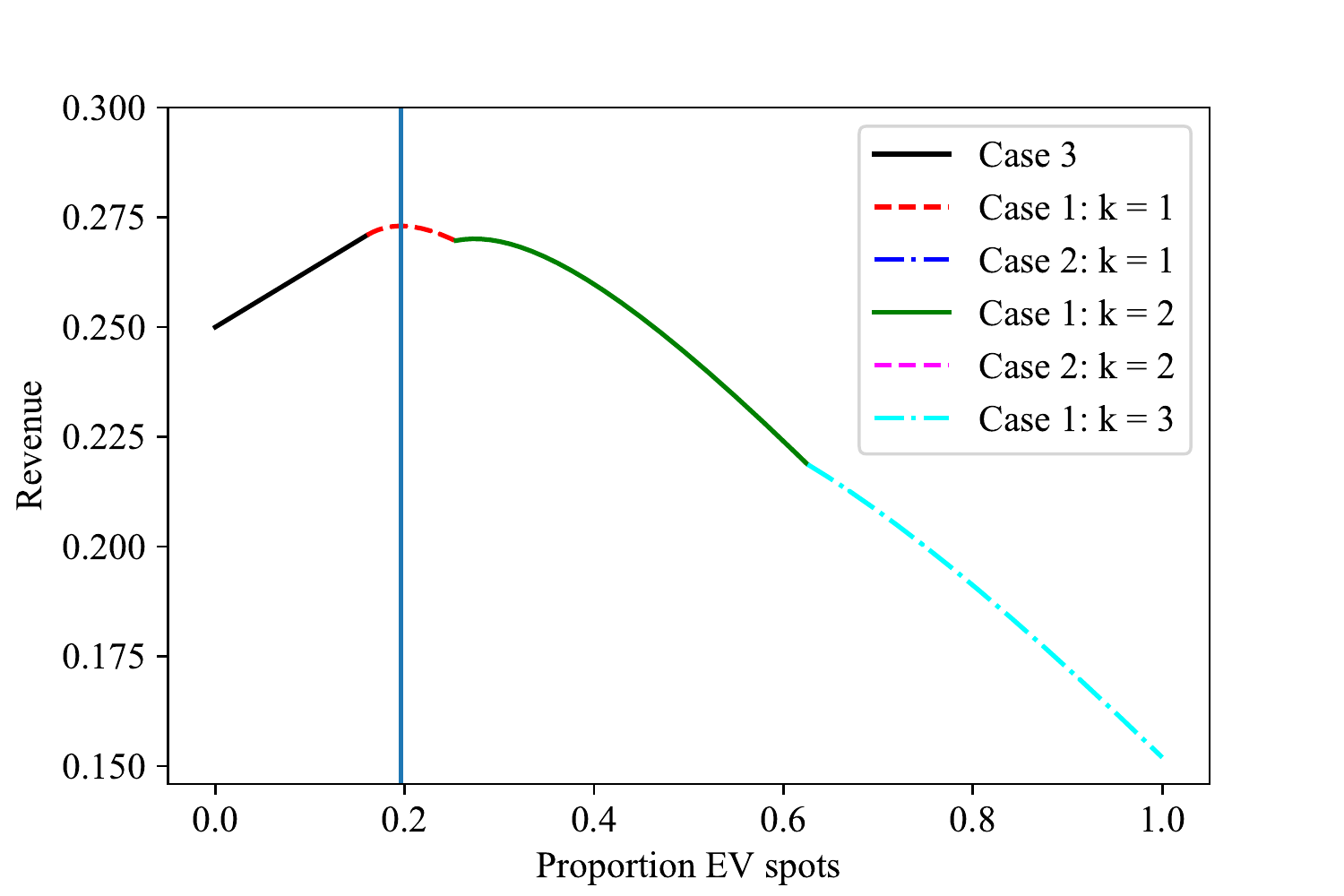} 
				\caption{$\mathbf{q} = [0.1,0.15,0.3]$, $\boldsymbol{\pi} = [0.4,0.33,0.27]$, $p=0.01$} 
				\label{fig:stochastic_plot_a}
			\end{subfigure}
		\end{minipage}
		\hspace{1.1cm}
		\begin{minipage}[t]{4cm}
			\begin{subfigure}[t]{1\textwidth}
				\centering	
				\includegraphics[scale=0.35]{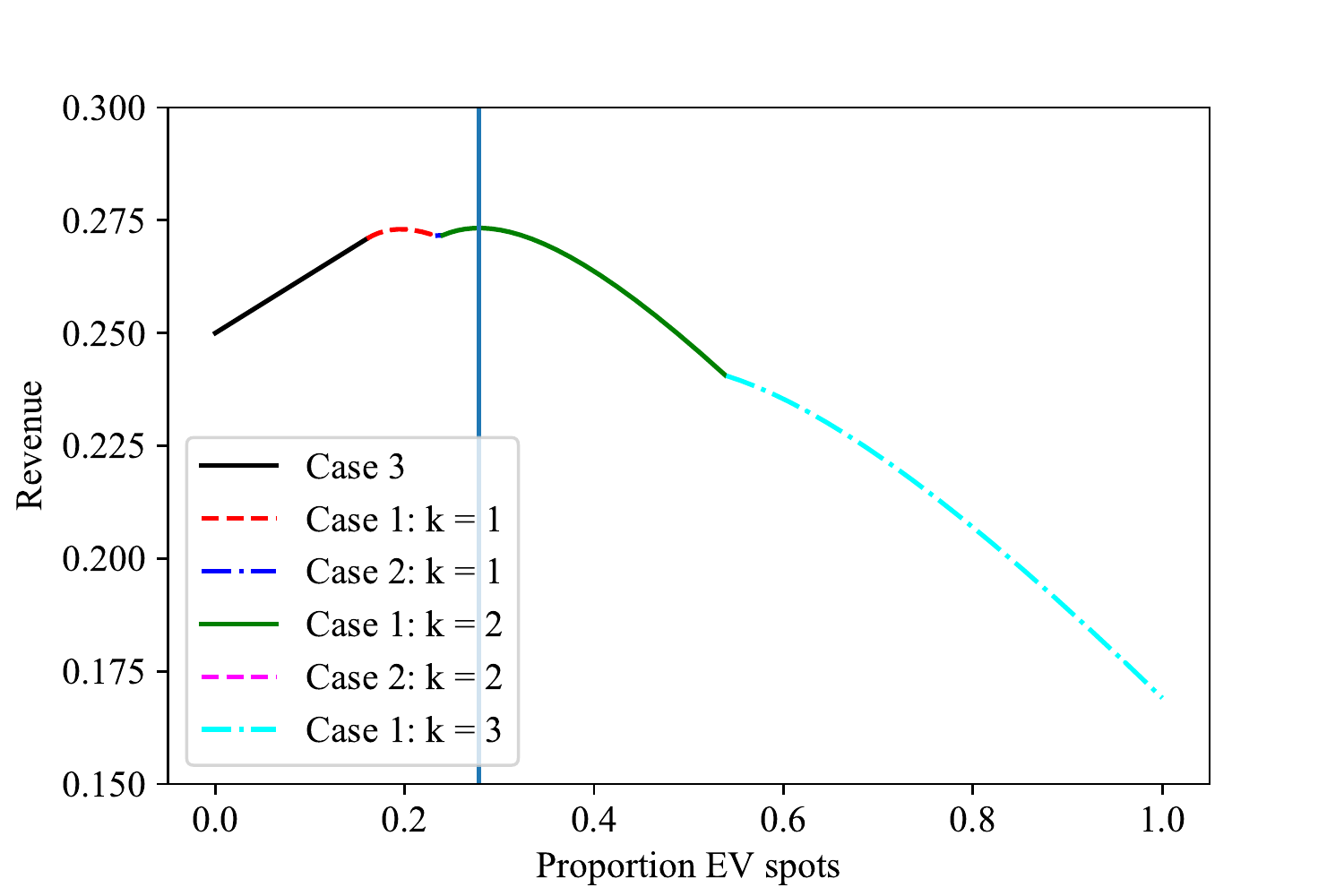} 
				\caption{$\mathbf{q} = [0.1,0.15,0.3]$, $\boldsymbol{\pi} = [0.31,0.33,0.36]$, $p=0.01$} 
				\label{fig:stochastic_plot_b}
			\end{subfigure}
		\end{minipage}
		\hspace{1.1cm}
		\begin{minipage}[t]{4cm}
			\begin{subfigure}[t]{1\textwidth}
				\centering	
				\includegraphics[scale=0.35]{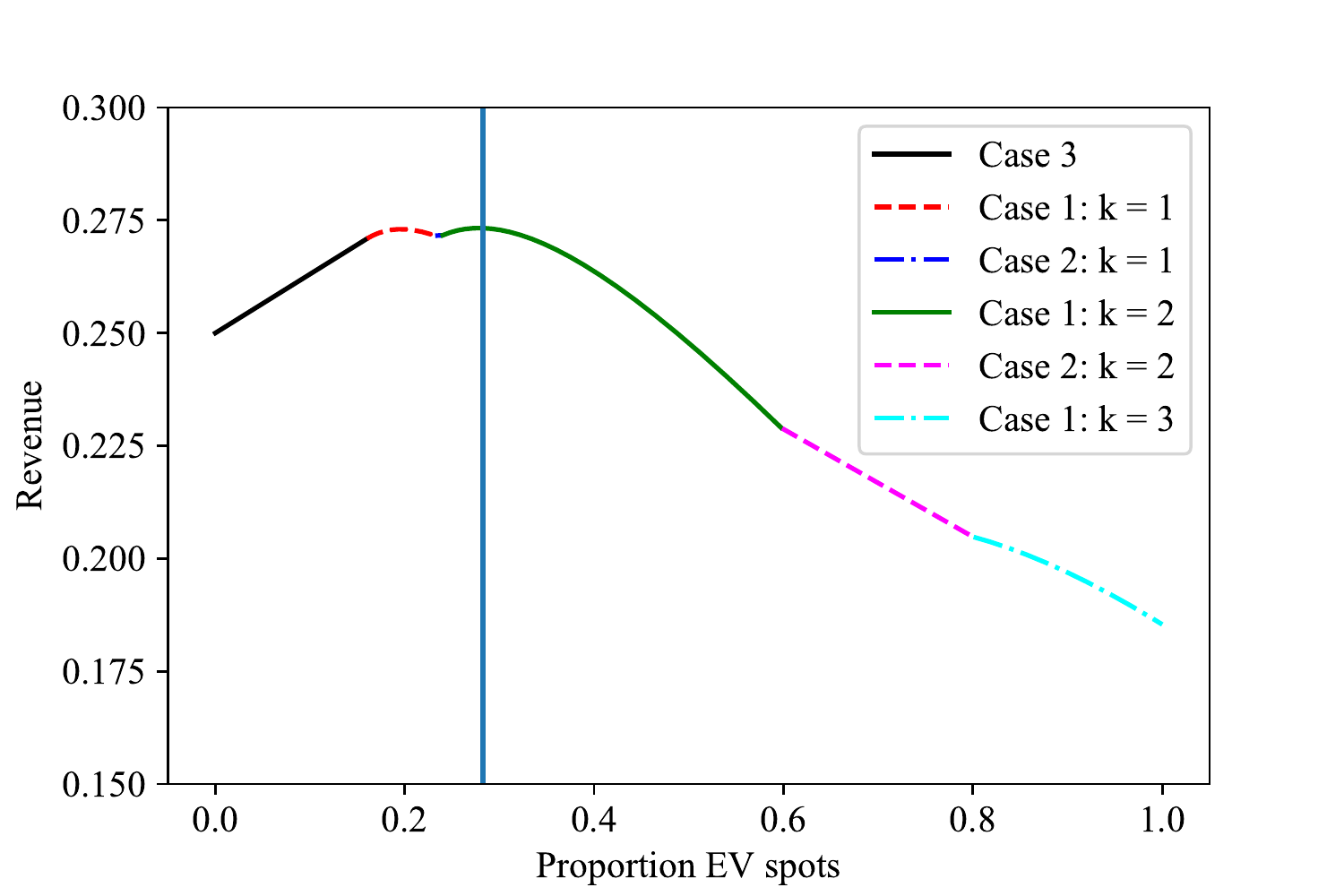} 
				\caption{$\mathbf{q} = [0.1,0.15,0.5]$, $\boldsymbol{\pi} = [0.31,0.33,0.36]$, $p=0.01$} 
				\label{fig:stochastic_plot_c}
			\end{subfigure}
		\end{minipage}
		
		\begin{minipage}[t]{4cm}
			\begin{subfigure}[t]{1\textwidth}
				\centering	
				\includegraphics[scale=0.35]{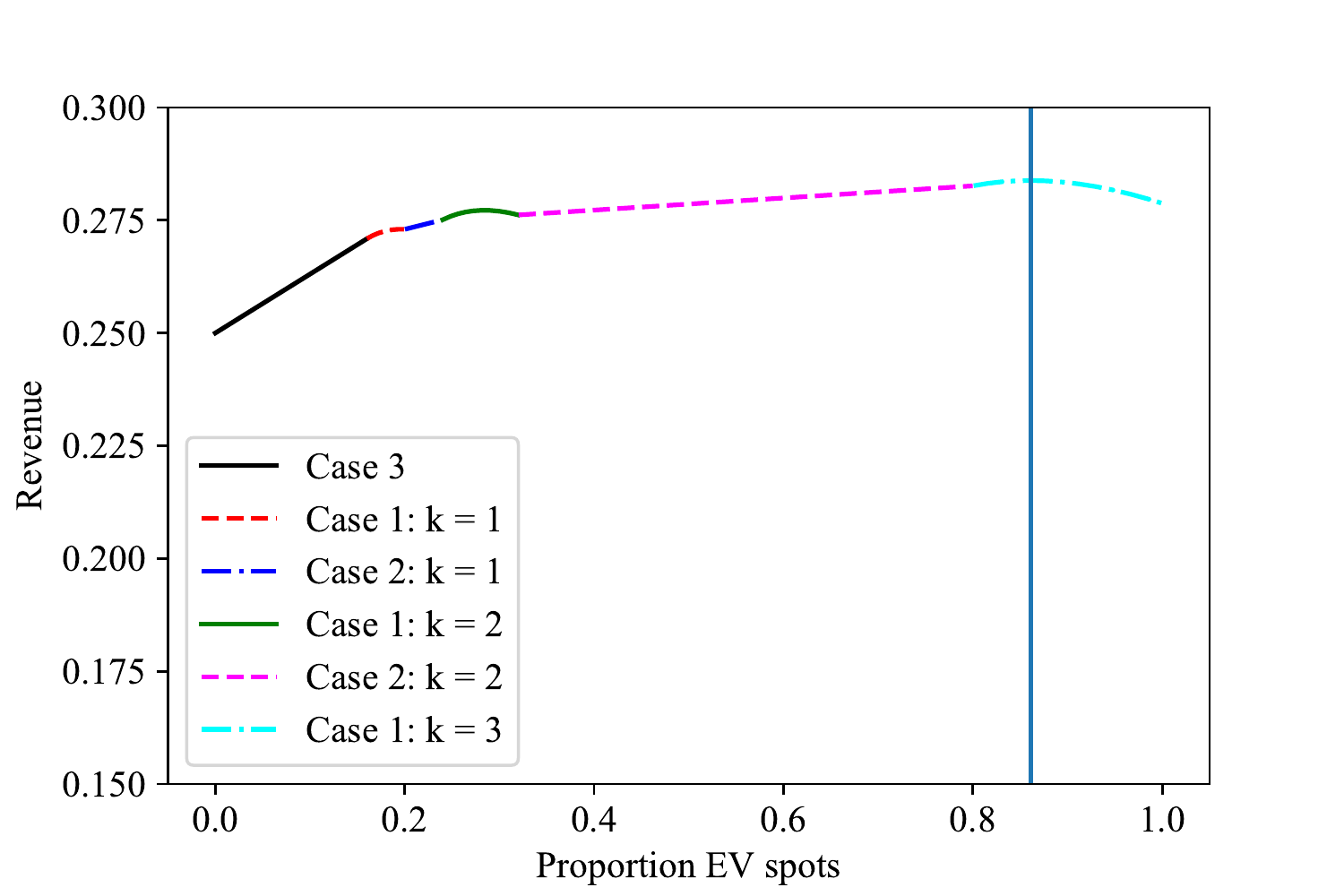} 
				\caption{$\mathbf{q} = [0.1,0.15,0.5]$, $\boldsymbol{\pi} = [0.2,0.1,0.7]$, $p=0.01$} 
				\label{fig:stochastic_plot_d}
			\end{subfigure}
		\end{minipage}
		\hspace{1.1cm}
		\begin{minipage}[t]{4cm}
			\begin{subfigure}[t]{1\textwidth}
				\centering	
				\includegraphics[scale=0.35]{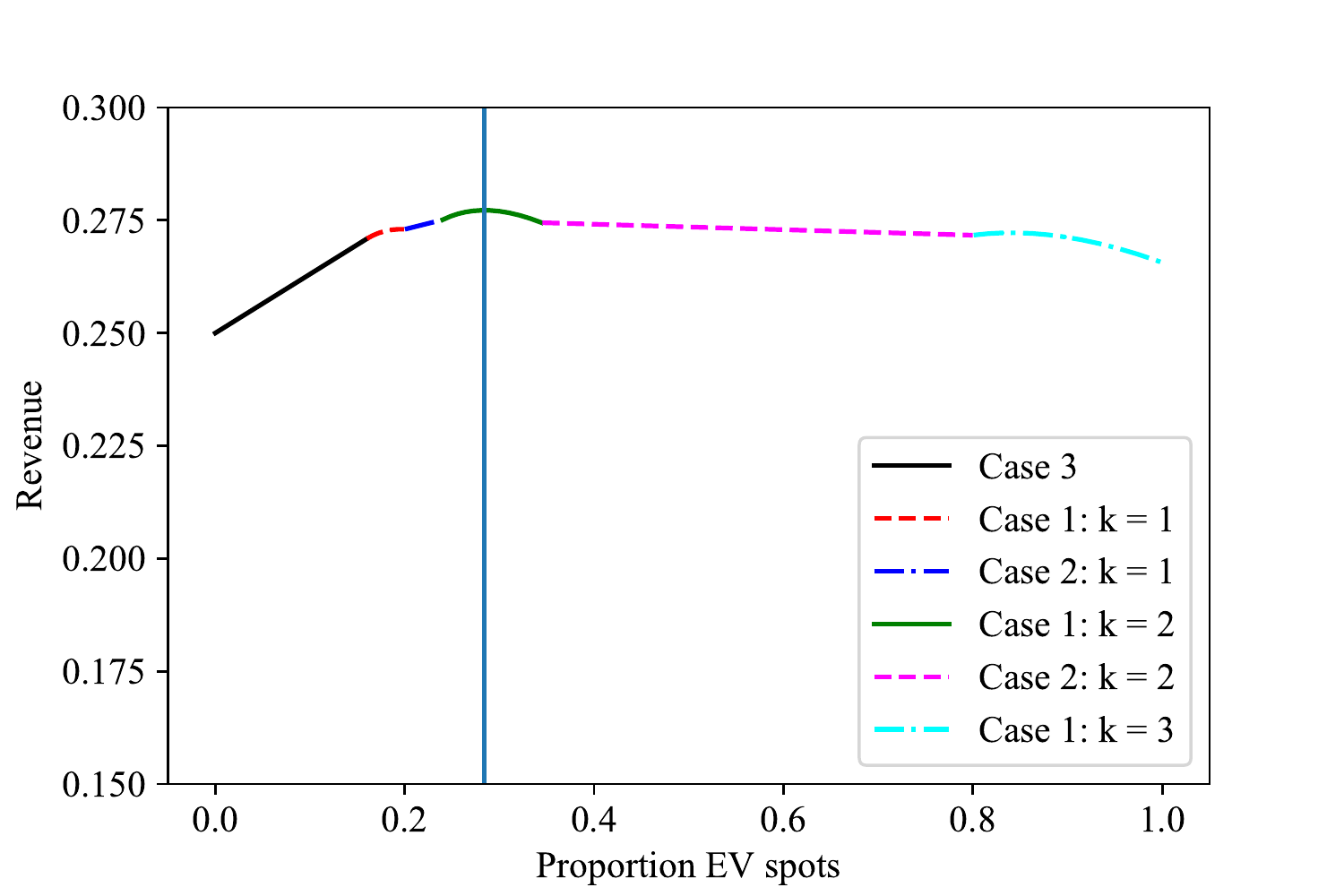}
				\caption{$\mathbf{q} = [0.1,0.15,0.5]$, $\boldsymbol{\pi} = [0.2,0.15,0.65]$, $p=0.01$} 
				\label{fig:stochastic_plot_e}
			\end{subfigure}
		\end{minipage}
		\hspace{1.1cm}
		\begin{minipage}[t]{4cm}
			\begin{subfigure}[t]{1\textwidth}
				\centering	
				\includegraphics[scale=0.35]{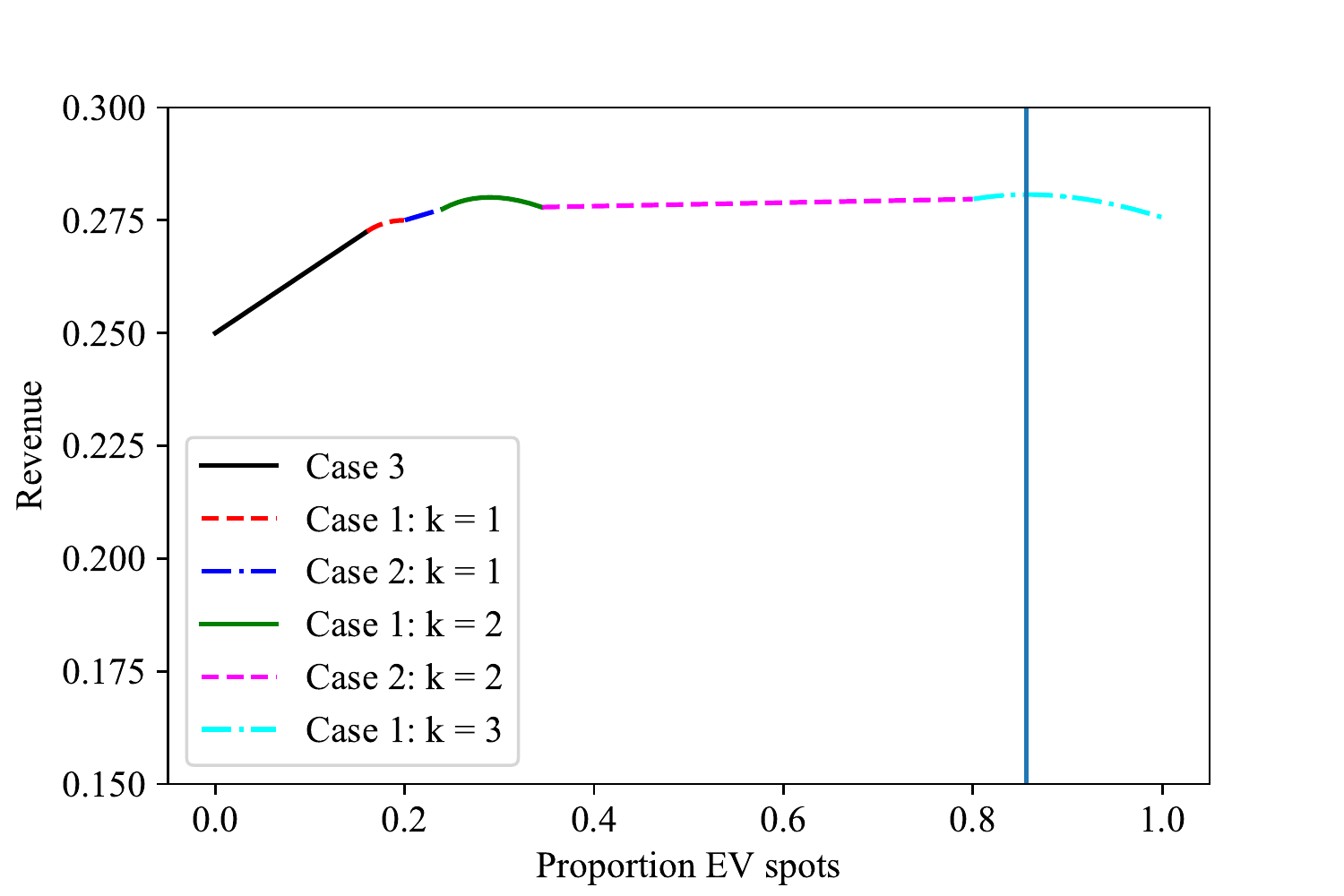}
				\caption{$\mathbf{q} = [0.1,0.15,0.5]$, $\boldsymbol{\pi} = [0.2,0.15,0.65]$, $p=0$} 
				\label{fig:stochastic_plot_f}
			\end{subfigure}
		\end{minipage}
		
		\caption{Plots for stochastic demand} 
		
	\end{figure*}
	
	\section{Monopolist with Stochastic Demand: Case Study} \label{sec:monopolistnumerical}
	
	We now examine how changes in the monopolist's belief about the EV market can affect the capacity they can add. In Figure 4 we plot the optimal revenue function for the monopolist as a function of $N_e$ with $W_e = 1.25$, $W_d = 1$, $\epsilon = 1$, $n=3$, $t = 0.01$ and $p=0.01$. We divide the pricing into the cases defined by different k values as described before. The overall optimal $N_e^*$ in each plot is marked by a vertical blue line. In each plot, we vary the distribution of EV arrivals as noted below. As indicated in Theorem \ref{thm:stochasticdemand}, in each plot the optimal number of spots is in a Case 1 pricing regime, but the target quantity $k$ is different and noted.
	
	In plot (a), we set $\mathbf{q} = [0.1,0.15,0.3]$ and $\boldsymbol{\pi} = [0.4,0.33,0.27]$ and find $N_e^* = 0.196$ with $k = 1$. In plot (b), we set $\mathbf{q} = [0.1,0.15,0.3]$ and $\boldsymbol{\pi} = [0.31,0.33,0.36]$, and find that the optimal exists at $N_e^* = 0.279$ with $k = 2$. This is the same as plot (a) but with a small change in the arrival distribution (a $0.06$ decrease in $\pi_1$ and a corresponding increase in $\pi_3$). This small change leads to an almost $50\%$ increase in the number of EV spots constructed. In plot (c), we set $\mathbf{q} = [0.1,0.15,0.5]$ and $\boldsymbol{\pi} = [0.31,0.33,0.36]$ and find $N_e^* = 0.278$ for k=2 again. This is the same as plot (b) but with $q_3$ increasing by $0.2$. We can see that even though there is now a possible realization where many more EV drivers enter the market than in plot (b), the monopolist optimally does not plan any differently as $N_e^*$ is independent of $q_3$. 
	
	In plot (d), we set $\mathbf{q} = [0.1,0.15,0.5]$ and $\boldsymbol{\pi} = [0.2,0.1,0.7]$ and find $N_e^* = 0.860$ with k=3. A change in the distribution of EV arrivals (not the quantities in each state) leads to a large change in the number of spots constructed. In plot (e), we set $\mathbf{q} = [0.1,0.15,0.5]$ and $\boldsymbol{\pi} = [0.2,0.15,0.65]$ and find $N_e^* = 0.284$ with $k=2$. This is a small change from plot d (the same quantity distribution but with $\pi_2$ $0.05$ higher and $pi_3$ correspondingly smaller). A small change in the distribution of arrivals in this case leads to more than a 3 times larger number of EV spots constructed. Lastly, in plot (f) we have the government subsidize the construction of EV spots ($p = 0.0$) with no change in the distribution of arrivals and find $N_e^* =0.857$ with k=3. We see that government subsidies can have a large impact on equilibrium behavior as it can shift which case is optimal for the monopolist. It also affects the number of spots within cases as we can see from \eqref{eq:optimalNecase1stochastic}. This makes it a tool that always increases the number of EV spots created, unlike the distributions of arrivals (which as we saw may not impact the number of spots constructed).
		
	\section{Conclusions} \label{sec:conclusion}
	
	We proposed a model where EV drivers and ICE drivers utilize the same fixed quantity of a resource (parking) and EV drivers have a higher valuation due to the need to charge. However, to serve EV drivers, operators require additional investment leading to these spots only being available to ECs. We have used this model to examine a case where there is competition or another where there is uncertainty in future EV demand. The right set of policies (subsidies and mandates) and the correct levels are dependent on how competitive the market is when firms are more cognizant of the ability to add capacity and price these two concurrent markets (EV and ICE drivers). With a single monopolist, we showed that errors in forecasting the EV demand can have significant impact on the charging capacity installed, suggesting that techniques for improving such forecasts are worth considering in future work. This also suggests that risk neutrality may not be the best way to optimize over this uncertainty, as a different objective function could lead to a capacity that targets for multiple sizes of the EV market instead of the one target as in our work.
	
	One key element in the EV market we did not model is the cost of electricity (we assumed charging costs were low enough to be considered 0); adding such costs as well as options to help lower these (like investing in renewable generation) is another possible future direction. 
	
	%%%%%%%%%%%%%%%%%%%%%%%%%%%%%%%%%%%%%%%%%%%%%%%%%%%%%%%%%%%%%%%%%%%%%%%%%%%%%%%%
	
		\section*{Appendix}
	
	\subsection{Justification of Wardrop equilibrium expressions}
	
	In \eqref{eq:wardropEV} and \eqref{eq:wardropICE}, we gave the Wardrop equilibrium quantities by assuming that both firms always served each type of user.  In particular we ignored the case where one firms prices for a given class was large enough compared to the other firms so that it did not serve any customers.  In this section we show that in terms of determining the Wardrop equilibrium, there is no loss in doing this.  In particular, we show that given an equilibrium as in Section \ref{sec:competitivemodel}, a firm would never benefit by unilaterally lowering its price to the point where our given expressions no longer hold. 
	
	We show this by contradiction. We analyze the EV market, but the same results hold for the ICE one. Without loss of generality, we will fix the price of firm 2 and look at firm 1's incentives to deviate. We note that the expression for quantity in \eqref{eq:wardropEV} we have implies
	\begin{equation*}
	c_1 + \epsilon \frac{q_{e1}}{Ne_{e1}} = c_2  + \epsilon \frac{q_{e2}}{Ne_{e2}}.
	\end{equation*}
	Assume that firm 1 deviates from this point and sets a price that leads to no consumers optimally parking at firm 2. This implies two conditions hold	
	\begin{equation} \label{eq:wardropcondition1}
	c_1  + \epsilon \frac{q_{e1}}{Ne_{e1}}  = W_e - \alpha W_e q_{e1}
	\end{equation}
	\begin{equation} \label{eq:wardropcondition2}
	c_2  > W_e - \alpha W_e q_{e1}.
	\end{equation}
	Rearranging equation \eqref{eq:wardropcondition1}, the revenue function for firm 1 is
	\begin{equation*}
	\max_{c_1} c_1(\frac{W_e - c_1}{\alpha W_e + \frac{\epsilon}{N_{e1}}}).
	\end{equation*}
	This is a concave function in $c_1$. Evaluating the first-order optimality conditions, the optimal price that firm 1 chooses in this case is	
	\begin{equation*}
	c_1^* = \frac{W_e}{2}.
	\end{equation*}
	In order for condition \eqref{eq:wardropcondition2} to hold given this price, it must be that:
	\begin{equation*}
	c_2 > W_e - \alpha W_e q_{e1} = c_1 + \epsilon \frac{q_{e1}}{N_{e1}} = \frac{W_e}{2} + \epsilon \frac{q_{e1}}{N_{e1}}.
	\end{equation*}	
	As $\frac{W_e}{2}$ is the monopolist price, this can never occur. Therefore, firm 1 will never best respond to firm 2 such that firm 2 would serve no consumers (unless firm 2's capacity was 0). 
	
	\subsection{Proof of Proposition \ref{prop:equilibriumexistence}}
	
	\subsubsection{Quantity case 1: Exactly follow mandate}
	
	In this case the first stage of the game is simply $Ne_i = rN_i$. Thus to prove this result we only need to verify that the second stage equilibrium is unique. To do this we first evaluate the second derivative of the revenue function from serving ICE drivers as follows (as quantity is assumed to be fixed in the first stage): \\
	\begin{gather*}
	\Pi^{ICE}_i = m_iq_{di} = \frac{(W_d+\beta W_d N_{dj}^{\epsilon}m_j)m_i}{\beta W_d(1+ \frac{N_{dj}}{N_{di}}) + \frac{1}{N_{di}^{\epsilon}}}  \\
	- \frac{(W_d \beta N_{dj}^{\epsilon} + 1)m_i^2}{\beta W_d(1+ \frac{N_{dj}}{N_{di}}) + \frac{1}{N_{di}^{\epsilon}}} \\
	\frac{\delta \Pi^{ICE}_i }{\delta m_i} = \frac{(W_d+\beta W_d N_{dj}^{\epsilon}m_j) - 2(W_d \beta N_{dj}^{\epsilon} + 1)m_i}{\beta W_d(1+ \frac{N_{dj}}{N_{di}}) + \frac{1}{N_{di}^{\epsilon}}} \\
	\frac{\delta^2 \Pi^{ICE}_i }{\delta m_i^2} = \frac{ - 2(W_d \beta N_{dj}^{\epsilon} + 1)}{\beta W_d(1+ \frac{N_{dj}}{N_{di}}) + \frac{1}{N_{di}^{\epsilon}}} \leq 0 .
	\end{gather*}
	
	Given the second derivative is always negative with respect to the firm's strategy we have the existence of a Nash Equilibria. Uniqueness then follows from comparing this to the off-diagonal terms in the Hessian matrix for $\Pi_i$, i.e.,  \\
	\begin{equation*}
	\frac{\delta^2 \Pi_i }{\delta m_i m_j} = \frac{W_d \beta N_{dj}^{\epsilon} }{\beta W_d(1+ \frac{N_{dj}}{N_{di}}) + \frac{1}{N_{di}^{\epsilon}}} \leq |\frac{\delta^2 \Pi_i }{\delta m_i^2}|. 
	\end{equation*}
	The same can be done with the revenue function of servicing EV drivers. This implies dominance solvability of the overall revenue function (as the Hessian of the payoff functions are negative semi-definite), which from Moulin (1984) proves there exists a unique Nash equilibrium \cite{moulin1984dominance} \cite{rosen1965existence}.  \\
	
	\subsubsection{Quantity case 2: Mandate forms lower bound} 
	
	For length purposes, the proof for this can be found in the version of the paper located at \cite{dropbox}.
	
	\subsection{Proof of Proposition \ref{prop:maxquantitydriver}}
	
	We assume $N_{e1} + N_{e2} = K$, and so can define $N_{e1} = K - N_{e2}$. We will prove this by showing the function $\frac{\delta q_{e}(K-N_{e2},N_{e2})}{\delta N_{e2}} \leq 0$ for $N_{e2} \geq \frac{K}{2}$. 
	
	Using \ref{eq:wardropEV} and the expressions for
	$c_i^{BR}(N_{ei},c_j)$ we can express $q_ei$ as follows: 
	\begin{equation*}
	q_{e1}(K-N_{e2},N_{e2}) = \frac{W_e + \frac{\alpha W_e N_{e2}}{\epsilon}c^{BR}_2(K-N{e2},N_{e2})}{2(\alpha W_e K + \epsilon)(\frac{1}{N_{e2}})},
	\end{equation*}
	\begin{equation*}
	q_{e2(K-N_{e2},N_{e2})} = \frac{W_e + \frac{\alpha W_e (K - N_{e2})}{\epsilon}c^{BR}_1(K-N{e2},N_{e2})}{2(\alpha W_e K + \epsilon)(\frac{1}{K - N_{e2}})}.
	\end{equation*}
	Let $q_e(K-N_{e2},N_{e2}) = q_{e1}(K-N_{e2},N_{e2}) + q_{e2}(K-N_{e2},N_{e2})$ and
	\begin{multline*}
	q_e(K-N_{e2},N_{e2}) =\frac{\alpha W_e(KN_{e2}-N^2_{e2})}{\epsilon}[ \\
	c^{BR}_1(K-N{e2},N_{e2}) + c^{BR}_2(K-N{e2},N_{e2})].
	\end{multline*}
	Using the expressions for $c^{BR}_1(K-N{e2},N_{e2})$ and $c^{BR}_2(K-N{e2},N_{e2})$ we find
	\begin{multline*}
	q_e(K-N_{e2},N_{e2}) = \frac{(\frac{\alpha W_e(KN_{e2}-N^2_{e2})}{\epsilon})(\frac{3\alpha W_e^2 K}{\epsilon} + 4W_e)}{\frac{3\alpha^2 W_e^2 (KN_{e2} - N_{e2}^2)}{\epsilon^2} + \frac{4 \alpha W_e K}{\epsilon} + 4}.
	\end{multline*}
	
	Using the quotient rule, we can express the derivative with respect to $N_{e2}$ as:
	
	\begin{multline}\label{eq:derivativequantity}
	\frac{\delta q_e(K-N_{e2},N_{e2})}{\delta N_{e2}} = \frac{(\frac{\alpha W_e(K-2N_{e2})}{\epsilon})(\frac{3\alpha W_e^2 K}{\epsilon} + 4W_e)}{\frac{3\alpha^2 W_e^2 (KN_{e2} - N_{e2}^2)}{\epsilon^2} + \frac{4 \alpha W_e K}{\epsilon} + 4} \\
	+  \frac{(\frac{\alpha W_e(KN_{e2}-N^2_{e2})}{\epsilon})(\frac{3\alpha W_e^2 K}{\epsilon} + 4W_e)(\frac{3\alpha^2 W_e^2 (K - 2N_{e2})}{\epsilon^2})}{(\frac{3\alpha^2 W_e^2 (KN_{e2} - N_{e2}^2)}{\epsilon^2} + \frac{4 \alpha W_e K}{\epsilon} + 4)^2}.
	\end{multline}
	
	For any $N_{e2} \in [\frac{K}{2},K]$, $(K-2N_{e2}) \leq 0$ but  $(KN_{e2}-N_{e2}^2) \geq 0$. Therefore, both terms in the above equation are $\leq 0$ for any $N_{e2}$ between $\frac{K}{2}$ and $K$, and so the total EV quantity serviced decreases as $N_{e2}$ increases from $\frac{K}{2}$ to $K$. Therefore, the quantity is maximized when $N_{e2} = \frac{K}{2}$. We note this same proof holds for ICE pricing.
	
	\subsection{Proof of Theorem \ref{thm:stochasticdemand}}
	
	We show that the monopolist can never choose $N_e^{*}$ such that they are in Case 2 or Case 3 (for any target $q_t$),  but can do so when in Case 1. We examine each case separately to demonstrate this.
	\begin{proof} 
		\begin{proofpart} %1
			Case 1: $q_{e} = q_i$ if $q_{max} \leq q_t$, else $q_{e} = q_t$
			
			\normalfont We can note using $c_1^{*}(N_e,t)$ as defined in \eqref{eq:pricecase1} in \eqref{eq:firststagemonopolist} leads to the revenue function in Case 1 as a function of the target $k$ as
			\begin{multline*}
			\Pi_{1}(N_e,k) =(W_e - \epsilon \frac{q_t}{N_e})(\sum_{i = 1}^{i = t}\pi_i q_i + \sum_{j = t+1}^{j = n}\pi_j q_t)\\
			+ \frac{(1-N_e)W_d^2}{4\epsilon} - pN_e. 
			\end{multline*}
			\normalfont This is a concave problem in $N_e$. We find the optimal $N_e^*$ for this case given a target quantity $q_t$ is	
			\begin{equation} \label{eq:optimalNecase1stochastic}
			N_e^*(k) = \max\{\sqrt{\frac{\epsilon q_t(\sum_{i = 1}^{i = k}\pi_i q_i + \sum_{j = k+1}^{j = n}\pi_j q_t)}{\frac{W_d^2}{4\epsilon} + p}},1\}.
			\end{equation} 
			
		\end{proofpart}
		\begin{proofpart} %2
			Case 2: $q_{e} = q_i$ if $q_{max} \leq q_t$, else $q_{e} = q_r$ ($ q_t < q_r < q_{t+1}$)
			
			\normalfont Substituting $c_2^{*}(N_e)$ as defined in \eqref{eq:pricecase2} into \eqref{eq:firststagemonopolist}, we can examine the revenue as a function of $N_e$ in Case 2 for a given $k$ as
			\begin{multline*}
			\Pi_{2}(N_e,k) = (\frac{W_e}{2} + \frac{\epsilon \sum_{i=1}^{k}\pi_iq_i}{2\sum_{j=k+1}^{n}\pi_j N_e})(\sum_{i = 1}^{i = k}\pi_i q_i 
			+ \\ 
			\sum_{j = k+1}^{j = n}\pi_j N_e\frac{W_e - \frac{\epsilon \sum_{i=1}^{k}\pi_iq_i}{\sum_{j=k+1}^{n}\pi_j N_e}}{2\epsilon}) + \frac{(1-N_e)W_d^2}{4\epsilon}. 
			\end{multline*}
			
			\normalfont This problem is twice differentiable in $N_e$, and examining the second derivative we can see that
			
			\begin{equation*}
			\frac{\delta^2 \Pi_{3}(N_e,k)}{\delta N_e^2} = \frac{\epsilon(\sum_{i=1}^{k}\pi_iq_i)^2}{2\sum_{j=k+1}^{n}\pi_j N_e^3} > 0.\\
			\end{equation*}
			
			\normalfont Therefore this problem is strictly convex in $N_e$ for any $k$. This means the revenue will be maximized within this case when one side of the constraint of \eqref{eq:case2twostageoptimization} is tight, which leads to Case 2 pricing no longer being feasible for the target quantity.
		\end{proofpart}
		\begin{proofpart} %1
			Case 3: $q_{e} < q_1 \ \forall q_{max} $
			
			\normalfont We can show the monopolist never optimally prices in this case by first noting the optimal price $c_3^{*}(N_e)$ defined in \eqref{eq:optmonopolistEVpricecasea} leads to an EV demand with $q_{max} = q_1$ as a function of $N_e$ equal to
			
			\begin{equation*} 
			D_e(N_e,c_3^{*}(N_e),q_1) = \min\{N_e \frac{W_e}{2\epsilon},q_1\}.
			\end{equation*}
			
			\normalfont We have assumed $\frac{W_e}{2} > \epsilon q_1$, which is equivalent to $D_e(N_e = 1,c = \frac{W_e}{2},q_{max}=q_1)$. This implies $\exists N_e^{'} < 1$. such that $N_e^{'} \frac{W_e}{2\epsilon} = q_1$ and thus at this level of $N_e^{'}$ the monopolist is pricing in Case 1 with $q_t = q_1$.
			
			\normalfont Secondly, we can substitute \eqref{eq:optmonopolistEVpricecasea} into \eqref{eq:firststagemonopolist} to get the revenue as a function of $N_e$ when they price in Case 3
			
			\begin{equation} \label{eq:revcase3}
			\Pi_3(N_e) = \frac{N_eW_e^2}{4\epsilon} + \frac{(1 - N_e)W_d^2}{4\epsilon} - pN_e.
			\end{equation}
			
			\normalfont We can show that at any $N_e < N_e^{'}$ (where the firm is still pricing in Case 3), $\Pi_1(N_e) < \Pi_1(N_e^{'})$ by examining the first derivative of \eqref{eq:revcase3}, which is
			
			\begin{equation*}
			\frac{\delta \Pi_1(N_e)}{\delta N_e} = \frac{4}{\epsilon}(W_e^2 - W_d^2) - p > 0
			\end{equation*}
			
			\noindent \normalfont as in our assumption $\frac{4}{\epsilon}(W_e^2 - W_d^2)  > p$. Therefore, at any level $N_e$ such that the constraint in \eqref{eq:case3optimization} is satisfied the monopolist would increase profit by increasing $N_e$, and at any level where the constraint is violated they would gain more profit from pricing in Case 1. 
		\end{proofpart}
	\end{proof}

	%%%%%%%%%%%%%%%%%%%%%%%%%%%%%%%%%%%%%%%%%%%%%%%%%%%%%%%%%%%%%%%%%%%%%%%%%%%%%%%%
	
	\begin{IEEEbiography}[{\includegraphics[width=1in,height=1.25in,clip,keepaspectratio]{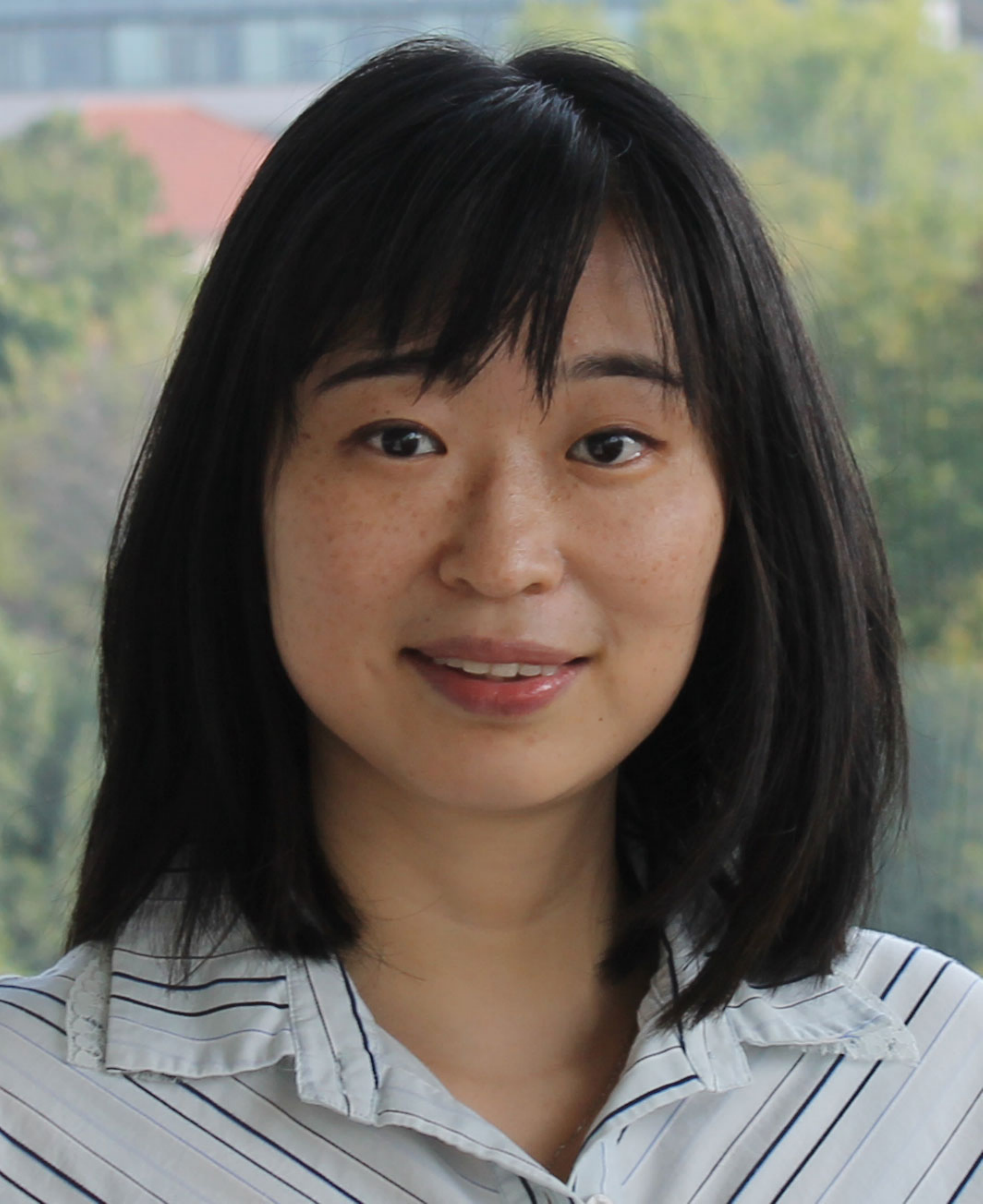}}]{Ermin Wei}
		
		received the Undergraduate triple degree in computer engineering, finance and mathematics with a minor in German from the University of Maryland, College Park, MD, USA, the M.S. degree from the Massachusetts Institute of Technology (MIT), Cambridge, MA, USA, and the Ph.D. degree in electrical engineering and computer science, advised by Prof. A. Ozdaglar, from MIT in 2014. She is currently an Assistant Professor with the Department of Electrical Engineering and Computer Science, Northwestern University, Evanston, IL, USA.  Her research interests include distributed optimization methods, convex optimization and analysis, smart grid, communication systems and energy networks, and market economic analysis. She was the recipient of many awards, including the Graduate Women of Excellence Award, second place prize in Ernst A. Guillemen Thesis Award, and Alpha Lambda Delta National Academic Honor Society Betty Jo Budson Fellowship.
	\end{IEEEbiography}
	
	\begin{IEEEbiography} [{\includegraphics[width=1in,height=1.25in,clip,keepaspectratio]{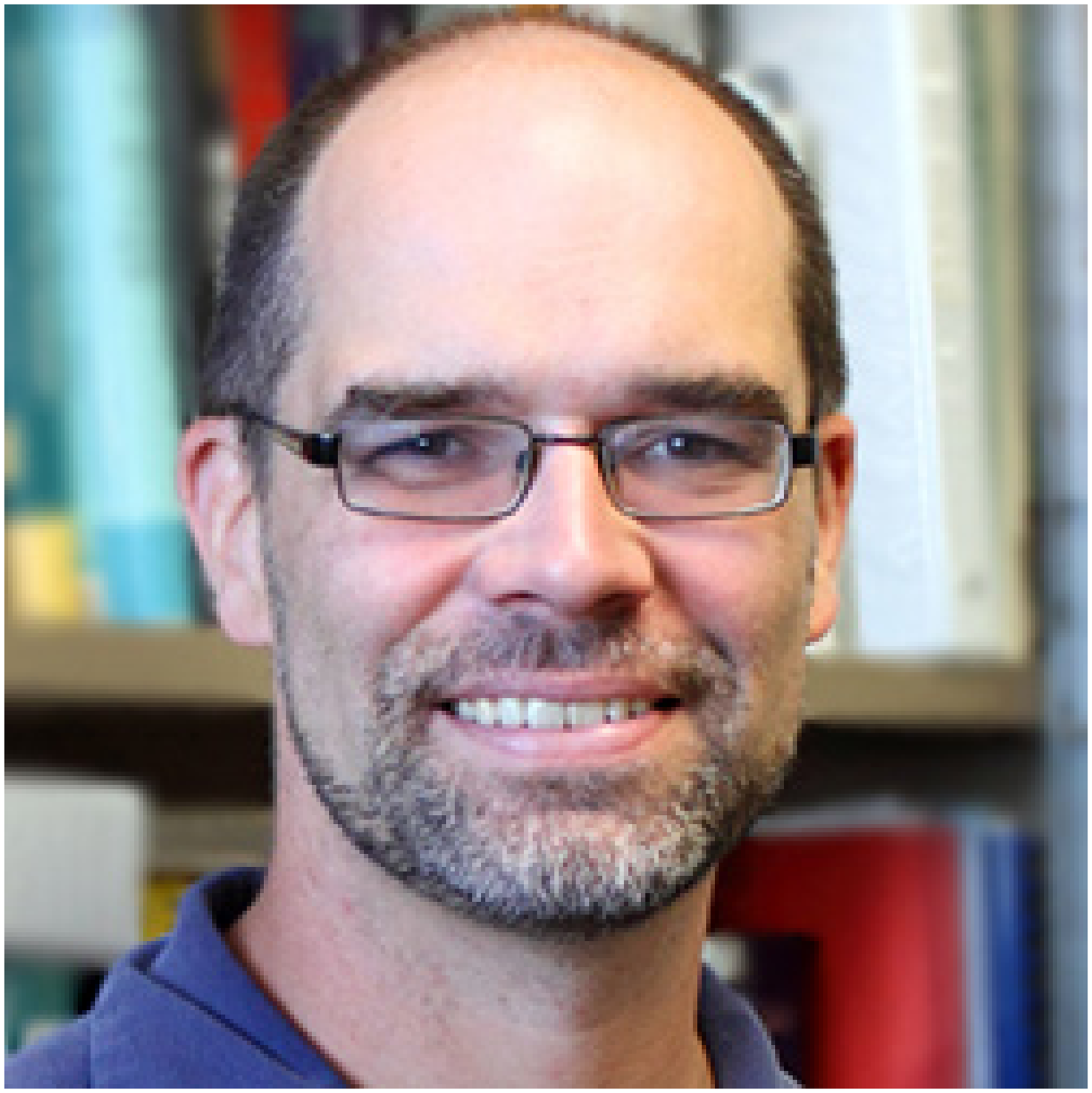}}]{Randall A. Berry} 
		
		Randall A. Berry is the Lorraine Morton Professor in the Department of Electrical Engineering and Computer Science at Northwestern University. His current research interests span topics in
		network economics, wireless communication, computer networking and information theory. Dr. Berry received the MS and PhD degrees from the Massachusetts Institute of Technology in 1996 and 2000, respectively. He is the recipient of a 2003 CAREER award from the National Science Foundation and is a fellow of the IEEE. He served as an editor for the IEEE Transactions on Wireless Communications (2006- 2009) and  the IEEE Transactions on Information Theory (2008-2012) as well as a guest editor for special issues of the IEEE Journal on Selected Areas in Communications (2017), the IEEE Journal on Selected Topics in Signal Processing (2008) and the IEEE Transactions on Information Theory (2007). Dr. Berry has also served on the program and organizing committees of numerous conferences including serving as co- chair of the 2012 IEEE Communication Theory Workshop, co-chair of 2010 IEEE ICC Wireless Networking Symposium. and TPC co-chair for the 2018 ACM MobiHoc conference.
	\end{IEEEbiography}
	
	\begin{IEEEbiography}
		[{\includegraphics[width=1in,height=1.25in,clip,keepaspectratio, angle=360]{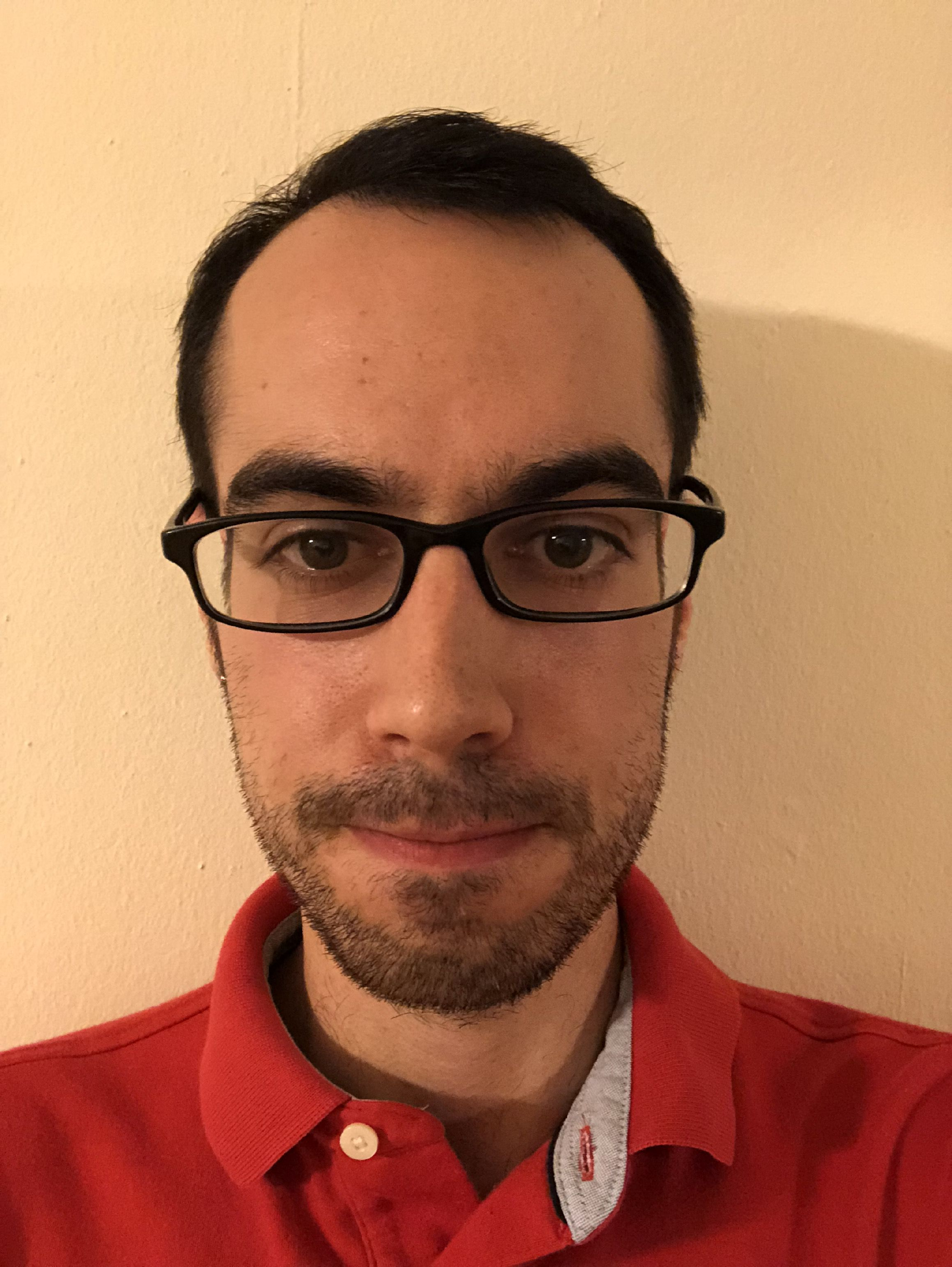}}]{Brendan Badia} 
		
		is a PhD candidate in Electrical Engineering at Northwestern. He received the BS in Electrical and Computer Engineering and  Economics at Carnegie Mellon, Pittsburgh, PA. 
	\end{IEEEbiography}
	
	%%%%%%%%%%%%%%%%%%%%%%%%%%%%%%%%%%%%%%%%%%%%%%%%%%%%%%%%%%%%%%%%%%%%%%%%%%%%%%%%

	%%%%%%%%%%%%%%%%%%%%%%%%%%%%%%%%%%%%%%%%%%%%%%%%%%%%%%%%%%%%%%%%%%%%%%%%%%%%%%%%

	\bibliographystyle{ieeetr}
	%\bibliography{Economics_Modern_Networks_Manuscript} 

\begin{thebibliography}{10}
		
		\bibitem{newsroom}
		"https://newsroom.aaa.com/2018/05/1-in-5-us-drivers-want-electric-vehicle/".
		
		\bibitem{seattlemag}
		"https://www.seattlemag.com/news-and-features/new-seattle-program-will-set-aside-parking-spots-throughout-city-electric-vehicle".
		
		\bibitem{dropbox}
		"https://www.dropbox.com/sh/846lljssdyif9c6/AADYOMkjt8EPH3owS-VSjhrha?dl=0".
		
		\bibitem{badia2018price}
		B.~Badia, R.~Berry, and E.~Wei,``Price and Capacity Competition for EV parking with Government Mandates,'' 
		{\em 2018 56th Annual Allerton Conference on Communication, Control, and Computing (Allerton)}, pp.~583--589, IEEE, 2018.
		
		\bibitem{lee2015electric}
		W.~Lee, L.~Xiang, R.~Schober, and V.~W. Wong, ``Electric vehicle charging
		stations with renewable power generators: A game theoretical analysis,'' {\em
			IEEE transactions on smart grid}, vol.~6, no.~2, pp.~608--617, 2015.
		
		\bibitem{mehar2013optimization}
		S.~Mehar and S.~M. Senouci, ``An optimization location scheme for electric
		charging stations,'' {\em IEEE SaCoNet}, pp.~1--5, 2013.
		
		\bibitem{luo2017placement}
		C.~Luo, Y.-F. Huang, and V.~Gupta, ``Placement of ev charging
		stations—balancing benefits among multiple entities,'' {\em IEEE
			Transactions on Smart Grid}, vol.~8, no.~2, pp.~759--768, 2017.
		
		\bibitem{alizadeh2016optimal}
		M.~Alizadeh, H.-T. Wai, M.~Chowdhury, A.~Goldsmith, A.~Scaglione, and
		T.~Javidi, ``Optimal pricing to manage electric vehicles in coupled power and
		transportation networks,'' {\em IEEE Transactions on Control of Network
			Systems}, 2016.
		
		\bibitem{liu2013optimal}
		Z.~Liu, F.~Wen, and G.~Ledwich, ``Optimal planning of electric-vehicle charging
		stations in distribution systems,'' {\em IEEE Transactions on Power
			Delivery}, vol.~28, no.~1, pp.~102--110, 2013.
		
		\bibitem{arnott2006integrated}
		R.~Arnott and E.~Inci, ``An integrated model of downtown parking and traffic
		congestion,'' {\em Journal of Urban Economics}, vol.~60, no.~3, pp.~418--442,
		2006.
		
		\bibitem{nguyen2016cost}
		T.~Nguyen, H.~Zhou, R.~A. Berry, M.~L. Honig, and R.~Vohra, ``The cost of free
		spectrum,'' {\em Operations Research}, vol.~64, no.~6, pp.~1217--1229, 2016.
		
		\bibitem{acemoglu2007competition}
		D.~Acemoglu and A.~Ozdaglar, ``Competition and efficiency in congested
		markets,'' {\em Mathematics of Operations Research}, vol.~32, no.~1,
		pp.~1--31, 2007.
		
		\bibitem{gibbens2000internet}
		R.~Gibbens, R.~Mason, and R.~Steinberg, ``Internet service classes under
		competition,'' {\em IEEE Journal on Selected Areas in Communications},
		vol.~18, no.~12, pp.~2490--2498, 2000.
		
		\bibitem{acemoglu2006price}
		D.~Acemoglu, K.~Bimpikis, and A.~Ozdaglar, ``Price and capacity competition,''
		tech. rep., National Bureau of Economic Research, 2006.
		
		\bibitem{kelly1998rate}
		F.~P. Kelly, A.~K. Maulloo, and D.~K. Tan, ``Rate control for communication
		networks: shadow prices, proportional fairness and stability,'' {\em Journal
			of the Operational Research society}, vol.~49, no.~3, pp.~237--252, 1998.
		
		\bibitem{fan2011distributed}
		Z.~Fan, ``Distributed demand response and user adaptation in smart grids,'' in
		{\em Integrated Network Management (IM), 2011 IFIP/IEEE International
			Symposium on}, pp.~726--729, IEEE, 2011.
		
		\bibitem{johari2010congestible}
		R.~Johari and S.~Kumar, ``Congestible services and network effects,'' {\em
			EC}, vol.~10, pp.~93--94, 2010.
		
		\bibitem{niyato2008market}
		D. ~Niyato and E. ~Hossain, ``Market-equilibrium, competitive, and cooperative pricing for spectrum sharing in cognitive radio networks: Analysis and comparison,'' in
		{\em IEEE Transactions on Wireless Communications}, vol.~11, no.~7, pp.~4273--4283, IEEE, 2008
		
		\bibitem{duan2012duopoly}  
		L.~ Duan, J.~ Huang, and B.~ Shou, ``Duopoly competition in dynamic spectrum leasing and pricing,'' in
		{\em IEEE Transactions on Mobile Computing}, vol.~11, no.~11, pp.~1706--1719, IEEE, 2012
		
		\bibitem{duan2016capacity}  
		L.~ Duan, B.~ Shou and J.~ Huang, ``Capacity allocation and pricing strategies for new wireless services,'' in
		{\em Production and Operations Management}, vol.~25, no.~5, pp.~866--882, Wiley Online Library, 2016
		
		\bibitem{glazer2001parking}
		A.~Glazer and E.~Niskanen, ``Parking fees and congestion,'' 2001.
		
		\bibitem{arnott1991temporal}
		R.~Arnott, A.~De~Palma, and R.~Lindsey, ``A temporal and spatial equilibrium
		analysis of commuter parking,'' {\em Journal of public economics}, vol.~45,
		no.~3, pp.~301--335, 1991.
		
		\bibitem{qian2013optimal}
		Z.~S. Qian and R.~Rajagopal, ``Optimal parking pricing in general networks with
		provision of occupancy information,'' {\em Procedia-Social and Behavioral
			Sciences}, vol.~80, pp.~779--805, 2013.
		
		\bibitem{shoup2006cruising}
		D.~C. Shoup, ``Cruising for parking,'' {\em Transport Policy}, vol.~13, no.~6,
		pp.~479--486, 2006.
		
		\bibitem{shoup1999trouble}
		D.~C. Shoup, ``The trouble with minimum parking requirements,'' {\em
			Transportation Research Part A: Policy and Practice}, vol.~33, no.~7-8,
		pp.~549--574, 1999.
		
		\bibitem{ayala2012pricing}
		D.~Ayala, O.~Wolfson, B.~Xu, B.~DasGupta, and J.~Lin, ``Pricing of parking for
		congestion reduction,'' in {\em Proceedings of the 20th International
			Conference on Advances in Geographic Information Systems}, pp.~43--51, ACM,
		2012.
		
		\bibitem{arnott1995modeling}
		R.~Arnott and J.~Rowse, ``Modeling parking,'' {\em Working Papers in
			Economics}, p.~282, 1995.
		
		\bibitem{lindsey2000traffic}
		C.~R. Lindsey and E.~T. Verhoef, ``Traffic congestion and congestion pricing,''
		tech. rep., Tinbergen Institute Discussion Paper, 2000.
		
		\bibitem{wardrop1952some}
		J.~G. Wardrop, ``Some theoretical aspects of road traffic research,'' in {\em
			Inst Civil Engineers Proc London/UK/}, 1952.
		
		\bibitem{haurie1985relationship}
		Haurie, Alain, and Patrice Marcotte. ``On the relationship between Nash--Cournot and Wardrop equilibria." 
		Networks 15.3 (1985): 295-308.
		
		\bibitem{plugincarschargenetworks}
		"http://www.plugincars.com/ultimate-guide-electric-car-charging-networks-126530.html".
		
		\bibitem{herfindahl1950concentration}
		O.~C. Herfindahl, ``Concentration in the us steel industry,'' {\em Unpublished
			PhD. Dissertation, Columbia University}, 1950.
		
		\bibitem{moulin1984dominance}
		H.~Moulin, ``Dominance solvability and cournot stability,'' {\em Mathematical
			social sciences}, vol.~7, no.~1, pp.~83--102, 1984.
		
		\bibitem{rosen1965existence}
		J.~B. Rosen, ``Existence and uniqueness of equilibrium points for concave
		n-person games,'' {\em Econometrica: Journal of the Econometric Society},
		pp.~520--534, 1965.
		
	\end{thebibliography}

\end{document}